\documentclass[12pt,a4paper,fleqn]{article}
\usepackage[doublespacing]{setspace}

\usepackage{amssymb,amsmath,amsthm,amsfonts,eurosym,geometry,ulem,graphicx,caption,color,setspace,sectsty,comment,footmisc,caption,natbib,pdflscape,subfigure,array,url,float,booktabs,multirow,colortbl,xcolor,bm,ragged2e,rotating}

\usepackage{hyperref} 

\newtheorem{assumption}{Assumption} 
\newtheorem{thm}{Theorem}
\theoremstyle{definition}
\newtheorem{lem}{Lemma}
\newtheorem{rem}{Remark}

\newtheorem{defn}{Definition}

\hypersetup{colorlinks=true,
            linkcolor=cyan,
            urlcolor=cyan,
            citecolor=cyan,
            hypertexnames=false,
            }
\usepackage{lmodern}
\usepackage{tgadventor,pifont, longtable,arydshln,booktabs}
\usepackage{fancyhdr}
\usepackage[flushleft]{threeparttable}
\usepackage{caption}
\usepackage{subcaption}

\title{The causal interpretation of panel vector autoregressions}
\author{Raimondo Pala\thanks{University of Rome Tor Vergata. Email: raimondopala@gmail.com \textit{Please note that this is a preliminary version. Future versions may include additional results.}} } 
\date{\today}

\begin{document}
	\maketitle
	\begin{abstract}
This paper discusses the different contemporaneous causal interpretations
of Panel Vector Autoregressions (PVAR). I show that the interpretation
of PVARs depends on the distribution of the causing variable, and
can range from average treatment effects, to average causal responses,
to a combination of the two. 

If the researcher is willing to postulate a no residual autocorrelation
assumption, and some units can be thought of as controls, PVAR can
identify average treatment effects on the treated. This method complements
the toolkits already present in the literature, such as staggered-DiD,
or LP-DiD, as it formulates assumptions in the residuals, and not
in the outcome variables. Such a method features a notable advantage:
it allows units to be ``sparsely'' treated, capturing the impact
of interventions on the innovation component of the outcome variables.

I provide an example related to the evaluation of the effects of
natural disasters economic activity at the weekly frequency in the
US.I conclude by discussing solutions to potential violations of the SUTVA assumption arising from interference.
\end{abstract}

\indent {\bf{JEL Classification}}:  C32, C33. \\ 
\noindent {\bf{Keywords}}: Vector autoregressive models, panel data, identification, difference in differences.    

\onehalfspacing
\thispagestyle{empty}

\pagebreak{}

\section{Introduction\label{sec:Introduction}}

Panel Vector Autoregressions (PVARs) are widely used across the social sciences because they provide an intuitive framework for estimating impulse response functions in large panel datasets. However, their use for causal inference remains limited, typically confined to notions of Granger causality. Recent advances at the intersection of time-series and causal inference show that traditional estimators can, under suitable conditions, admit causal interpretations in the Rubin causal framework (\citealt{bojinovshephard2019, menchettibojinov2022}).

This paper extends that insight to PVARs. I show that PVARs possess a causal interpretation under the same independence conditions discussed in \citet{RambachanSheppard2021}, and I characterize how the identified causal estimands, such as an Average Treatment Effect (ATE) and Average Causal Response (ACR) depend on the distribution of the policy variable.

The most important contribution of this paper is to show that Panel Vector Autoregressions (PVARs) can also be used to conduct causal inference in the presence of endogenous policies. By exploiting their panel structure, the estimation of a PVAR can be reframed as a dynamic comparison between treated and control units. In this sense, PVARs can serve as a useful and complementary tool to Difference-in-Differences (DiD) methods (\citealt{Card1990,CardKrueger}), as well as to their recent extensions such as staggered DiD (\citealt{sunabraham2021}) and local-projection DiD (\citealt{dube2023local}).

The key intuition is that if the autoregressive component of the model is capable of absorbing the dynamic effects of treatment, then previously treated units can act as valid controls in subsequent periods. This makes it possible to identify causal effects through the recursive estimation of a PVAR without the need for explicit orthogonality or independence assumptions typically required in other frameworks.

This approach also highlights several advantages over conventional DiD estimators. Standard DiD designs usually assume that once a unit is treated, it remains treated forever, restricting the set of possible control observations. Local-projection DiD estimators allow units to re-enter the control group after a certain time, but they impose a specific treatment duration ex ante.

PVARs offer a more flexible perspective. Because their identification focuses on the residuals rather than on mean outcome differences, it suffices to examine the distribution of the innovations to address residual autocorrelation. If treated and control units are comparable after controlling for the autoregressive component, the resulting contrasts can be interpreted as an Average Treatment effect on the Treated (ATT). Importantly, this ATT does not arise from direct comparisons of observed outcomes but from differences in the innovation components of treated and untreated units.

The paper concludes with an application related to the evaluation of the effects of natural disasters and discussing possible solutions to violations of the SUTVA assumption.

Narrative examples are constantly used through the paper to give a
proper meaning to the theoretical results, and the discussion is organized
as follows. Section \ref{sec:Introduction-to-Panel} discusses a causal
framework for PVARs. Section \ref{sec:Causal-effects} defines different
types of causal effects.

Section \ref{sec:Estimation-of-causal} shows that PVARs can estimate
an ATE if the policy is a dummy, an ACR if the policy is continuous,
a mix of the two if it is continuous and non-negative. Section \ref{subsec:policy-heterog}
shows that PVARs can estimate ATT if some units are extracted as controls
and there is no residual autocorrelation. Section \ref{sec:Applications}
provides an empirical example of this last result by discussing
the causal effect of natural disasters on weekly economic activity
in the US. Section \ref{sec:interference} discusses non-compliance.

Section \ref{sec:Conclusion} concludes.

\section{Introduction to Panel Vector Autoregressions \label{sec:Introduction-to-Panel}}

Consider a number of known interventions and outcomes, respectively
indexed by $k=1,...,K$ which could indicate fiscal and monetary policy,
and $j=1,..,J$ which could indicate output and employment. A PVAR
is generally ran on the aggregation of processes 
\[
x_{it}=(W_{k=1,it}^{\prime},W_{k=2,it}^{\prime},...,W_{k=K,it}^{\prime},Y_{j=1,it}^{\prime},Y_{j=2,it}^{\prime},..,Y_{j=J,it}^{\prime})',
\]
 where $W_{k,it}$ is a variable that indicates the treatment value
for unit $i$ by a treatment of type $k$ at time $t$, $Y_{j,it}$
is the outcome variable $j$ for unit $i$ at time $t$. 

Panel Vector Autoregressions are generally represented as a process
that depends on its past, a series of unit-specific characteristics,
and some random disturbances, hence:
\begin{equation}
\begin{aligned}x_{it}=(I_{m}-\Phi)\mu_{i}+\Phi x_{i,t-1}+\tilde{x}_{it}\qquad & i=1,..,N\quad t=1,..,T\end{aligned}
\label{eq:main_PVAR}
\end{equation}
where $\Phi$ denotes an $m\times m$ matrix of slope coefficients\footnote{This framework could be extended to random effects instead of fixed
effects. However, this change would have no impact on the nature of
the causal effects estimated.}, $\mu_{i}$ is an $m\times1$ vector of individual-specific effects,
$\tilde{x}_{it}$ is an $m\times1$ vector of disturbances, and $I_{m}$
denotes the identity matrix of dimension $m\times m$. The model can
be extended to include higher lags, but to keep the notation compact
I will use a one lag representation. \footnote{Any VAR(p) for $p>1$ can always be represented as a VAR(1).}
Moreover, through the paper I will make the implicit assumptions that
$x_{it}$ is stationary and that the disturbances are assumed to be
multivariate iid distributions with mean zero and not necessarily
normal. This is because assuming the normality of the distributions
could be in opposition to some of the conclusions that will be drawn
on the paper.\footnote{I will show later that the normality of the distribution of the first
variable in the policy column of $\widetilde{x}_{it}$ is just one
case over many possible.}

The focus of the following section will be on the disturbances 
\[
\tilde{x}_{it}=(\widetilde{W}_{k=1,it}^{\prime},\widetilde{W}_{k=2,it}^{\prime},...,\widetilde{W}_{k=K,it}^{\prime},\widetilde{Y}_{j=1,it}^{\prime},\widetilde{Y}_{j=2,it}^{\prime},..,\widetilde{Y}_{j=J,it}^{\prime})',
\]
 where the tilde represents the specific disturbance related to the
original variable. I will assume that the policy variables go first,
and the outcome variables follow. This system is in line with the
conventional way of thinking about Vector Autoregressions. For this
reason there exists one causal effect that is estimated per $j=1,..,J$
that corresponds to every intervention variable $k=1,..,K.$
\begin{rem}
The results that follow can also be obtained from a different system
in which the PVAR only includes the outcome variables and which residuals
are regressed against a $W_{k,it}$ which has the same distribution
as $\widetilde{W}_{k,it}$. This system may present several advantages
in the case of a no autocorrelated intervention which include, but
are not limited to: (1) a better capability of testing the distribution
of $W_{k,it}$ than $\widetilde{W}_{k,it}$, (2) the ability to avoid
an overparametrisation.\footnote{Another way of viewing this issue could be to consider this system
as the panel equivalent to the one analysed by \citet{bojinovshephard2019}
for the case of AR and VAR estimators.}
\end{rem}
\begin{rem}
The results that follow will be valid only in the case of recursive
identification. This is because other methods of inference, such as
sign, short run and long run restrictions, are much more difficult
to give a causal interpretation if they only partially identify the
system.\footnote{This is a known issue in the macroeconomics literature which is a
leading cause to several debates around the interpretability of Impulse
Response Functions in VARs and how restrictive the restrictions can
and should be. For example, the oil shock literature has long discussed
about the appropriateness of the IRF generated by a VAR that is partially
identified because it is unclear whether the assumption that oil shocks
can only have effects within {[}-1.5,0{]} is reasonable in light of
Bayesian inference (\citet{KilianOil2009,kilianzhou2019,baumeisterhamilton2019}).
Moreover, sign restrictions have been shown to assume a boundary that
is frequently not discussed by the economist but is implicit in the
way the rotations are formalized (\citet{baumeisterhamilton2015econometrica}).} However, methods such as IV or proxy variables (\citealp{Mertens2014,Olea2021})
have been shown to posses a causal interpretation akin to the LATE
one by \citealp{RambachanSheppard2021}. A follow up to this paper
includes a full description of identification, estimation and inference
in PVARs identified through external instruments (\citet{pala2024identificationestimationPVARIV}).
\end{rem}

\section{Causal effects \label{sec:Causal-effects}}

The object of the estimation of any VAR is a \emph{dynamic causal
effect. }Such effects arise as the difference between the outcome
variables under the assignment path $\widetilde{w}_{i,1:t}\in\mathcal{W}^{t}$
and a counterfactual path $\widetilde{w}_{i,1:t}'\in W^{t}$. Hence,
the objective of the estimation is the difference $\widetilde{Y}_{it}(\widetilde{w}_{i,1:t})-\widetilde{Y}_{it}(\widetilde{w}_{i,1:t}')$.
A reasonable way to discuss the potential outcome process is to represent
it as a function of past assignments, potential contemporaneous assignments,
and future assignments (see \citealp{RambachanSheppard2021}):
\[
\widetilde{Y}_{it}(\widetilde{w}):=\widetilde{Y}_{it}(\widetilde{W}_{i,1:t-1},\widetilde{w},\widetilde{W}_{i,t+1:T}).
\]

When it comes to Panel VARs, the contemporaneous effect is usually
identified by imposing restrictions on the covariance matrix of the
residuals. The $t+h$ effect, for $h\geq1$, is instead generally
computed by the means of the Impulse Response Functions.\footnote{This paper will only refer to the causal interpretation of the former, as the interpretation of the dynamic IRF follows trivially.}

Starting from this definition, I will consider the following causal
effects:
\begin{defn}
(Causal effects).\label{def:(ATTs)} For $t\geq1$ , and any fixed
$\widetilde{w},\widetilde{w}'\in\mathcal{W}$
\end{defn}
\begin{enumerate}
\item \textbf{Average Treatment Effect} is $\mathbb{E}[\widetilde{Y}_{it}(\widetilde{w})-\widetilde{Y}_{it}(\widetilde{w}')]$, 
\item \textbf{Average Treatment effect on the Treated }is $\mathbb{E}[\widetilde{Y}_{it}(\widetilde{w})-\widetilde{Y}_{it}(\widetilde{w}^{\prime})|\widetilde{W}_{it}=\widetilde{w}]$, 
\item \textbf{Average Causal Response} is $\frac{\delta\mathbb{E}[\widetilde{Y}_{it}(\lambda)]}{\delta\lambda}$
,
\item \textbf{Average Causal Response on the Treated }is $\frac{\delta\mathbb{E}[\widetilde{Y}_{it}(\lambda)|\widetilde{W}_{it}=\lambda]}{\delta\lambda}$.
\end{enumerate}
The difference $\widetilde{Y}_{it}(\widetilde{w})-\widetilde{Y}_{it}(\widetilde{w}')$
is the so called Impulse effect. The Average Treatment Effect (ATE)
is usually the target of analysis, since it is informative of the
effect of receiving treatment $\widetilde{w}$ versus receiving the
treatment $\widetilde{w}'$. 

The average treatment effect on the treated (ATT) therefore measures
its expectation, but only conditionally on the subset of time being
one in which there is a treatment. It carries a substantially different
meaning compared to ATEs since it includes a selection bias term.\footnote{The selection bias is $\frac{cov(\widetilde{Y}_{j,it}(w_{k}),1\{\widetilde{W}_{k,it}=w_{k}\})}{var(1\{\widetilde{W}_{k,it}=w_{k}\})}$.}

Notice that both of these effects assume a structure associated to
the policy variable that is a dummy variable that is equal to $\tilde{w}$
for treated units and to $\tilde{w}^{\prime}$ for the untreated ones.

The definitions of the average causal response (ACR) is instead well
discussed by \citealp{callawaysantannabacon2021}. It represents the
causal effect of changes in the policy variables on the outcome variable.
In this context, the researcher could easily derive the effect of
a change of a given size in a policy on the outcome variable. Lastly,
the concept of average causal response on the treated (ACRT) represents
the same idea, but it includes the conditioning factor on the observed
times, indicating the presence of a selection bias component. Both
of those causal effects can result from continuous policy variables. 

The researcher is usually interested in the effects of changes in
one specific policy variable to an outcome variable. Hence, the dynamic
causal effects for a particular assignment $k$, defined $\widetilde{w}_{k}\in\mathcal{W}_{k}$
on the outcome $j$, are defined as:
\[
\widetilde{Y}_{j,it}(\widetilde{w}_{k}):=\widetilde{Y}_{j,it}(\widetilde{W}_{i,1:t-1},\widetilde{W}_{1:k-1,it},\widetilde{w}_{k},\widetilde{W}_{k+1:K,it},\widetilde{W}_{i,t+1:T}).
\]

\begin{enumerate}
\item Its associated \textbf{ATE} is $\mathbb{E}[\widetilde{Y}_{j,it}(\widetilde{w}_{k})-\widetilde{Y}_{j,it}(\widetilde{w}_{k}')]$,
\item its associated \textbf{ATT} is $\mathbb{E}[\widetilde{Y}_{j,it}(\widetilde{w}_{k})-\widetilde{Y}_{j,it}(\widetilde{w}_{k}')|\widetilde{W}_{k,it}=\widetilde{w}_{k}]$, 
\item its associated \textbf{ACR} is $\frac{\delta\mathbb{E}[\widetilde{Y}_{j,it}(\lambda_{k})]}{\delta\lambda_{k}}$,
\item its associated \textbf{ACRT} is $\frac{\delta\mathbb{E}[\widetilde{Y}_{j,it}(\lambda_{k})|\widetilde{W}_{k,it}=\lambda_{k}]}{\delta\lambda_{k}}$.
\end{enumerate}

\section{Estimation of causal effects of interest of Panel Vector Autoregressions
\label{sec:Estimation-of-causal}}

Until now, the structure of the policy variable has been kept as flexible
as possible to highlight the high potential of PVAR to estimate different
causal effects depending on the distribution of the policy. In the
following sections, I will show that PVAR can estimate several different
causal effects depending on the distribution of the policy variable.
Figure \ref{fig:Tree} represents the different type of causal effects
that can be identified according to the policy variable, which will
be discussed more in depth in the subparagraphs that follow.

\begin{figure}[h]
\centering{}\includegraphics[scale=0.8]{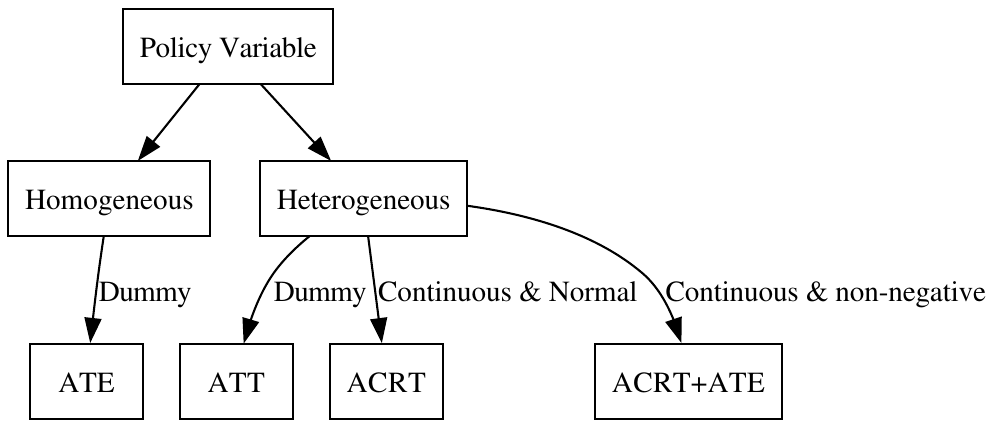}\caption{Causal effects depend on the structure of the policy variable. \label{fig:Tree}}
\end{figure}
Through the paper, the stable unit value treatment assumption is assumed to hold. In this context, it can be framed as follows
\begin{assumption}(SUTVA): \label{assu:SUTVA}
For each unit $i$ and time $t$, the potential outcome given a contemporaneous assignment 
$\widetilde{w}$ and past and future assignments of the unit,
\[
\widetilde{Y}_{j,it}(\widetilde{w}_{k}) := \widetilde{Y}_{it}\big(\widetilde{W}_{k,i,1:t-1}, \widetilde{w}_{k}, \widetilde{W}_{k,i,t+1:T} \big),
\]
is invariant to the assignment paths of all other units:
\[
\widetilde{Y}_{j,it}\big(\widetilde{W}_{k,i,1:t-1}, \widetilde{w}_{k}, \widetilde{W}_{k,i,t+1:T}; \widetilde{W}_{k,-i,1:T}\big)
= \widetilde{Y}_{j,it}\big(\widetilde{W}_{k,i,1:t-1}, \widetilde{w}_{k}, \widetilde{W}_{k,i,t+1:T}\big),
\]
where $\widetilde{W}_{k,-i,1:T} = \{\widetilde{W}_{k,i,1:T}: -i\neq i\}$ denotes the collection of all other units' assignment paths.
\end{assumption}
potential violations of the assumption due to interference are discussed in section \ref{sec:interference}, which revises the main empirical application explicitly estimating spillover effects.

\subsection{The policy variable is randomized and homogeneous across units\label{subsec:policy-time-dummy}}

A researcher is interested in the evaluation of the effects of hurricanes
in the economy of some states. All the data they have collected corresponds
to unemployment, inflation, and the moments in which all the states
have been affected by the disaster shock. This means that they have
no control regions available. The researcher decides to make the assumption
that, were the regions not affected by a disaster shock, they would
have followed a VAR process. In this case, all units are defined as
either treated or control in each period. 
\begin{assumption}
(Policy variable is homogeneous across units): \label{assu:(Policy-homogeneous)}
For all $i=1,..,N$, at a specific time $t$, either $\widetilde{W}_{k,it}=1$
or $\widetilde{W}_{k,it}=0$. 
\end{assumption}
\begin{assumption}
(Policy variable is random): \label{assu:(Policy-variable-is-random)}
The policy variable is independent from the other policies, its past,
future and the potential outcome process, so that:
\[
\widetilde{W}_{k,it}\perp(\widetilde{W}_{k,i,1:t-1},\widetilde{W}_{1:k-1,i,t},\widetilde{W}_{k+1:K,i,t},\widetilde{W}_{k,i,t+1:T},\{\widetilde{Y}_{ji,t:T}(1):1\in\mathcal{W}^{t:T}\})
\]

\noindent for all $k$, $j$, $i$, $t$ .
\end{assumption}
Assumption \ref{assu:(Policy-homogeneous)} indicates that the treated
moments are the ones for which $\widetilde{W}_{k,it}=1$,
while control units are the ones for which $\widetilde{W}_{k,it}=0$. This interpretation aligns with common practice in IRF analysis, where
 shocks are typically expressed as unit shocks or one-standard-deviation shocks.
Assumption \ref{assu:(Policy-variable-is-random)} indicates that
the hurricanes:
\begin{enumerate}
\item Must be independent from previous hurricanes (after controlling for
their autocorrelation).
\item Must be independent with respect to future hurricanes.
\item Must be independent with respect to the potential outcome of inflation
and unemployment.
\item Must be independent with respect to other influencing factors of inflation
and unemployment.
\end{enumerate}
If the previous assumptions are satisfied, obtaining the conditions
under which the estimated causal effects are ATEs is straightforward.
Notice that the estimated effects will be the following:
\begin{thm}
(PVAR: Randomized, dummy, homogeneous policy): \label{thm:(PVARs-Impulse-Response}Under
assumption \ref{assu:SUTVA} and \ref{assu:(Policy-homogeneous)}, Panel Vector Autoregressions
estimate, for all $j$, all $k$, all $i$ and all $t$, the following
causal effect:

\[
\gamma_{jk}=\mathbb{E}[\widetilde{Y}_{j,it}|\widetilde{W}_{k,it}=1]-\mathbb{E}[\widetilde{Y}_{j,it}|\widetilde{W}_{k,it}=0]
\]

\noindent Which means that $\gamma_{jk}$ can be decomposed into the
average treatment effect and a selection bias term as follows:

\[
\gamma_{jk}=\mathbb{E}[\widetilde{Y}_{j,it}(1)-\widetilde{Y}_{j,it}(0)]+\Delta_{j,k,i,t}(1,0)
\]

\noindent where:
\end{thm}
\[
\Delta_{j,k,i,t}(1,0):=\frac{Cov(\widetilde{Y}_{j,it}(1),1\{\widetilde{W}_{k,it}=1\})}{\mathbb{E}[1\{\widetilde{W}_{k,it}=1\}]}-\frac{Cov(\widetilde{Y}_{j,it}(0),1\{\widetilde{W}_{k,it}=0\})}{\mathbb{E}[1\{\widetilde{W}_{k,it}=0\}]}
\]
This means that the causal effects identified by the IRF of a PVAR
will be Average Treatment Effects only if $\Delta_{j,k,i,t}(1,0)=0$,
i.e. the selection bias is zero. This can happen if assumption \ref{assu:(Policy-variable-is-random)}
is satisfied.
\begin{thm}
(PVARs estimate ATEs): \label{thm:(PVARs-estimate-ATEs).} Under assumptions \ref{assu:SUTVA},
\ref{assu:(Policy-homogeneous)} and \ref{assu:(Policy-variable-is-random)}
the following is true:
\begin{equation}
\begin{aligned}Cov(\widetilde{Y}_{j,it}(1),1\{\widetilde{W}_{k,it}=1\})=0,\qquad & Cov(\widetilde{Y}_{j,it}(0),1\{\widetilde{W}_{k,it}=0\})=0\end{aligned}
\label{eq:pt1_of_theorem_ATT}
\end{equation}

\noindent then $\Delta_{j,k,it}(1,0)=0$.
Moreover, equation \ref{eq:pt1_of_theorem_ATT} is satisfied if: 
\[
\begin{aligned}\widetilde{W}_{k,it}\perp\widetilde{Y}_{j,it}(1), & \text{and} & \widetilde{W}_{k,it}\perp\widetilde{Y}_{j,it}(0)\end{aligned}
\]

\noindent which is in turn implied by:
\[
\widetilde{W}_{k,it}\perp\{\widetilde{Y}_{j,it}(1):1\in\mathcal{W}_{k}\}
\]

\noindent which can be satisfied if assumption \ref{assu:(Policy-variable-is-random)}
is respected. Hence, Panel Vector Autoregressions identify, for all
$j$, all $k$, all $t$ and $i$:
\[
\gamma_{jk}=\widetilde{\text{ATE}}_{j,k,i,t}=\mathbb{E}[\widetilde{Y}_{j,it}(1)-\widetilde{Y}_{j,it}(0)]
\]
\end{thm}

This result is similar to the one of \citealp{RambachanSheppard2021}
for VARs. The intuition for the similarity is that the system is purposefully
shrinked down to the case in which unit heterogeneity is not utilized
by the system. I will focus now in the cases in which the distribution
of the policy variable may not be binary. Subsection \ref{subsec:policy-continuous-normal}
discusses the case of a normal distribution of the policy variable's
innovations.

\subsection{The policy variable is continuous and normally distributed\label{subsec:policy-continuous-normal}}

A researcher is interested in evaluating the effects of monetary policy
on inflation, similarly to the problem of \citet{Sims1980}. Since
their interest is in a monetary union, they collect data on interest
rates, inflation, and other possibly relevant variables from all the
states belonging to the monetary union. They estimate a PVAR, and
observe that the innovation component of the interest rates is normally
distributed. They make the argument that after controlling for past
information, monetary policy decisions are as if random, and leverage
this assumption by using a recursively identified PVAR. To carry a
causal claim, they will need the following assumptions. 
\begin{assumption}
(Continuous differentiability): \label{assu:(Continuous-differentiability)}
The forecast errors of the outcome variables $\widetilde{Y}_{j,it}$
are continuously differentiable with respect to $\widetilde{W}_{k,it}.$
\end{assumption}
\begin{assumption}
(Normal distribution): \label{assu:(Normal-distribution)}The forecast
error of the policy variable is normally distributed $\widetilde{W}_{k,it}\sim\mathcal{N}(0,\sigma_{\widetilde{W}_{k,it}}^{2})$.
\end{assumption}
Assumption \ref{assu:(Continuous-differentiability)} may appear technical,
but indicates that the forecast errors of inflation are not ill defined
in some regions. Let us define $g(\widetilde{w}_{k})=\mathbb{E}[\widetilde{Y}_{j,it}|\widetilde{W}_{k,it}=\lambda_{k}]$. 
\begin{thm}
(PVARS: Randomized, continuous, unbounded Policy): \label{thm:(PVARS:-Randomised,-continuous,}Under
assumptions \ref{assu:(Continuous-differentiability)}, \ref{assu:(Normal-distribution)},
and for each $j$, $k$, $i$, and $t$, Panel Vector Autoregressions
capture the following estimand: 
\[
\gamma_{jk}=\int q(\lambda_{k})g'(\lambda_{k})d\lambda_{k}.
\]

\noindent Here $q(\lambda_{k})>0$ and $\int q(\lambda_{k})d\lambda_{k}=1$.
The weights are:
\[
\begin{aligned}q(\lambda_{k})=\frac{1}{\sigma_{\widetilde{W}_{k,it}}^{2}}\int_{-\infty}^{\infty}(\mathbb{E}[\widetilde{W}_{k,it}]F_{\widetilde{W}_{k,it}}(\lambda_{k})-\theta_{\widetilde{W}_{k,it}}(\lambda_{k})),\end{aligned}
\]

\noindent and $\Theta_{\widetilde{W}_{k,it}}(\lambda_{k})=\int_{-\infty}^{\lambda_{k}}mf_{\widetilde{W}_{k,it}}(m)dm=F_{\widetilde{W}_{k,it}}(\lambda_{k})\mathbb{E}[\widetilde{Y}_{j,it}|\tilde{W}_{k,it}=\lambda_{k}]$.
\end{thm}
Here $f_{\widetilde{W}_{k}}$ indicates the marginal density, $F_{\widetilde{W}_{k}}$
the marginal cumulative distribution, $\sigma_{\widetilde{W}_{kit}}^{2}$
indicates the variance of $\widetilde{W}_{k,it}$. This theorem indicates
that the causal effect estimated is a set of weights $q(\lambda_{k})$
that depends only on the distribution of interest rates disturbances.
\begin{thm}
(PVARs identify ACRT):\label{thm:(PVARs-estimate-the-ACRT)}Under
assumptions \ref{assu:SUTVA}\ref{assu:(Continuous-differentiability)} and \ref{assu:(Normal-distribution)},
Panel Vector Autoregressions identify, for all $j$, $k$, $i$, $t$:

\[
\gamma_{jk}=\widetilde{\text{ACRT}}_{j,k,it}(\lambda_{k})
\]
\end{thm}
This result is partially similar to the one of \citealp{Yitazaki1996}.
The ACRT is of particular interest since it captures the effect on
the forecast error of the outcome variable by moving along the curve
of observed forecast errors of the policy. This first derivative interpretation,
however, possess a selection bias in $g(\widetilde{w}_{k})$, which
depends on the conditioning argument $\widetilde{W}_{k,it}=\lambda_{k}$.
This is where the assumption that the monetary policy is fully exogenous
with respect to inflation can help obtain a causal effect without
a selection bias.
\begin{thm}
(PVARs identify ACR): \label{thm:(PVARs-identify-ACR):} Under assumptions \ref{assu:SUTVA},\ref{assu:(Policy-variable-is-random)},\ref{assu:(Continuous-differentiability)},
and \ref{assu:(Normal-distribution)}, 
\[
\gamma_{jk}=\widetilde{\text{ACR}}_{jk,it}(\lambda_{k}).
\]
\end{thm}
Notice that $\widetilde{\text{ACR}}_{jk,it}(\lambda_{k})$ is free
of any selection bias, and allows to formulate counterfactual policy
scenarios. If one truly believes in the assumption that monetary policy
innovations are fully random with respect to inflation, it becomes
possible to identify the effects of any impact shock in monetary policy
on inflation.

\subsection{The policy variable is continuous and non negative\label{subsec:policy-non-negative}}

A researcher is, once again, interested in the evaluation of the impact
of hurricanes in the economy of some states. This time they make the
argument that not all hurricanes impact the economy in the same way.
Some create more destruction; some simply affect mildly the economy;
some are technically registered, but have no effect on the economy.
Therefore, they collect the value of all the destroyed property. The
periods in which there are no hurricanes are registered as zero, and
the periods in which they happen, the destruction measure is instead
utilized. Therefore, the treatment variable is either positive or
zero. Such a restriction is the same one proposed by \citealp{callawaysantannabacon2021}
of either continuous or multi-valued treatments. 
\begin{assumption}
(Continuous differentiability): \label{assu:(Continuous-differentiability)-1}
The forecast errors of the outcome variables $\widetilde{Y}_{j,it}$
are continuously differentiable with respect to $\widetilde{W}_{k,it}$.
\end{assumption}
\begin{assumption}
(Non-negative $\widetilde{W}_{k,it}$):\label{assu:(Greater-than-zero)}
$\widetilde{W}_{k,it}$ satisfies $0\leq\widetilde{W}_{k,it}<\infty$
for all $i$ and all $t$. 
\end{assumption}
\begin{assumption}
(Strong Parallel Trends): \label{assu:(Strong-Parallel-Trend)} For
all $j$, $k$, $i$, $t$ 
\[
\mathbb{E}[\widetilde{Y}_{j,it}(\lambda_{k})-\widetilde{Y}_{j,it}(0)]=\mathbb{E}[\widetilde{Y}_{j,it}(\lambda_{k})-\widetilde{Y}_{j,it}(0)|\widetilde{W}_{k,it}=\lambda_{k}].
\]
\end{assumption}
Assumption \ref{assu:(Continuous-differentiability)-1} is akin to
the one in the previous subsection. Assumption \ref{assu:(Greater-than-zero)}
regularizes the limits of the distribution of the treatment variable.
Assumption \ref{assu:(Strong-Parallel-Trend)} is instead used by
\citealp{callawaysantannabacon2021} to restrict the treatment variable
to have the same potential effect for the treated units and control
units. Such an assumption imposes that receiving a treatment of size
$\lambda_{k}$ compared to receiving no treatment (a form of ATE)
is the same quantity as conditioning on the group of treated units.
It is a convenient, but strong, assumption that allows to freely move
from ATEs to ATTs. It is required to only hold for $\lambda_{k}$
and zero potential outcomes as this is the source of selection bias.

The causal effect can be computed from the Cholesky decomposition
of the estimated residuals, and will represent:
\begin{thm}
(PVARS: Randomized, continuous, non-negative Policy): \label{thm:(PVARs-estimate-weighted-average)}
Under assumption \ref{assu:SUTVA}\ref{assu:(Continuous-differentiability)-1}, Panel
Vector Autoregressions capture the following estimand, for all $j$,
all $k$, all $t$ and all $i$:
\[
\begin{aligned}\gamma_{jk}= & \int_{d_{L}}^{d_{U}}q_{1}(\lambda_{k})\left(\frac{\delta\mathbb{E}[\widetilde{Y}_{j,it}|\widetilde{W}_{k,it}=\lambda_{k}]}{d\lambda_{k}}+\frac{\mathbb{E}[\widetilde{Y}_{j,it}|\widetilde{W}_{k,it}=\lambda_{k}]-\mathbb{E}[\widetilde{Y}_{j,it}|\widetilde{W}_{k,it}=a]}{\delta a}\mid_{a=\lambda_{k}}\right)d\lambda_{k}+\\
 & +q_{0}\frac{\mathbb{E}[\widetilde{Y}_{j,it}|\widetilde{W}_{k,it}=d_{L}]-\mathbb{E}[\widetilde{Y}_{j,it}|\widetilde{W}_{k,it}=0]}{d_{L}}.
\end{aligned}
\]

\noindent Here 
\[
\frac{\mathbb{E}[\widetilde{Y}_{j,it}|\widetilde{W}_{k,it}=\lambda_{k}]-\mathbb{E}[\widetilde{Y}_{j,it}|\widetilde{W}_{k,it}=a]}{\delta a}\mid_{a=\lambda_{k}}
\]
is a selection bias term and
\[
\begin{aligned}q_{1}(\lambda):=\frac{\mathbb{E}[\widetilde{W}_{k,it}|\widetilde{W}_{k,it}\geq\lambda_{k}]-\mathbb{E}[\widetilde{W}_{k,it}])\mathbb{P}(\widetilde{W}_{k,it}\geq\lambda_{k})}{var(\widetilde{W}_{k,it})}\end{aligned}
\]

\noindent and
\[
q_{0}:=\frac{(\mathbb{E}[\widetilde{W}_{k,it}|\widetilde{W}_{k,it}>0]-\mathbb{E}[\widetilde{W}_{k,it}])\mathbb{P}(\widetilde{W}_{k,it}>0)d_{L}}{var(\widetilde{W}_{k,it})}
\]

\noindent and the weights satisfy $\int_{d_{L}}^{d_{U}}q_{1}(\lambda_{k})d\lambda_{k}+q_{0}=1$.

Under assumption \ref{assu:(Strong-Parallel-Trend)} it becomes
\[
\begin{aligned}\gamma_{jk}= & \int_{d_{L}}^{d_{U}}q_{1}(\lambda_{k})\frac{\delta\mathbb{E}[\widetilde{Y}_{j,it}|\widetilde{W}_{k,it}=\lambda_{k}]}{d\lambda_{k}}d\lambda_{k}+q_{0}\frac{\mathbb{E}[\widetilde{Y}_{j,it}|\widetilde{W}_{k,it}=d_{L}]-\mathbb{E}[\widetilde{Y}_{j,it}|\widetilde{W}_{k,it}=0]}{d_{L}}\end{aligned}
\]
\end{thm}
This approach decomposes the estimated coefficient in different partitions
of a curve of causal effects, so that $\gamma_{jk}$ will be a weighted
average of an $\text{ACRT}_{j,k,it}(\lambda)$ and an $\text{ATE}_{j,k,it}(d_{L})$.
\begin{thm}
(PVARs identify a weighted average of ACR and ATE): \label{thm:(PVARs-estimate-average-ATE-ACR)}
The result in Theorem \ref{thm:(PVARs-estimate-weighted-average)}
can be expressed as a weighted average of ACR and ATEs, hence Panel
Vector Autoregressions identify, for all $j$, $k$, $t$, $h$:

\[
\gamma_{jk}=\int_{d_{L}}^{d_{U}}q_{1}(\lambda_{k})\widetilde{\text{ACRT}}_{j,k,it}(\lambda_{k})+q_{0}\frac{\widetilde{\text{ATE}}_{j,k,it}(d_{L})}{d_{L}},
\]

\noindent which collapses to a different range of causal effects.
\end{thm}
This result can be interpreted as a negative one. Even under strong
assumptions, the asymmetry of the distribution can result in a mix
of causal effects, which may be of hard interpretation. Moreover,
the causal interpretation of $\widetilde{\text{ACRT}}(\lambda_{k})$
may still be problematic due to the presence of a selection bias term.
A key difference with respect to \citealp{callawaysantannabacon2021}
is that while they consider other possible forms of estimators for
quantities that may be informative about the shape of the causal response,
the approach of PVARs has an additional downside. The researcher has
no control over the distribution of the innovations of the policy
variable, making other alternative estimators that appropriately weight
the observations difficult to consider. 

What does this mean for the researcher that collected information
about the negative impact of hurricanes? Intuitively, this results
speaks about the exercise of trying to capture different types of
information within a unique coefficient. In this case, $\gamma_{jk}$
is being asked to capture both the impact of moving from a non-event
to the lowest destructive hurricane, but also the one of moving from
the lowest possible hurricane to the second lowest possible and so
on. This is not necessarily a unique coefficient. In the case of \citealp{callawaysantannabacon2021}
the consequence is that it becomes possible to move from a parametric
to a non-parametric estimation, and capture the different causal effects
through the estimation of different coefficients. 

This issue opens up several venues for research. While currently PVARs
are identified by simply running a recursive ordering based estimation,
it could be possible to consider different estimators that select
different observations due to their common distribution. In that case,
rather than obtaining just one Impulse Response Function, the researcher
would have in hand several Impulse Responses, each capturing the impact
of an intervention of, say, a treatment 1 versus a treatment 0; or
a treatment 2 versus a treatment 0; or a treatment 2 versus a treatment
1; each referring to a different Cholesky decomposition that refers
to commonly distributed residuals. Unfortunately such estimators generally
require a large amount of data to cover all the different densities,
which may be reasonable in the case of weekly or monthly macroeconomic
data, but unfeasible in the case of quarterly or yearly data. In general,
however, the issue of non-linearity is becoming more and more apparent
in many macroeconomic estimators (see \citet{kolesar2024dynamic}).

\section{The policy variable is a dummy that indicates treated or untreated
units\label{subsec:policy-heterog}}

A researcher is, for the last time, interested in the evaluation of
the impact of hurricanes to the economy of some states. This time
they observe that hurricanes affect only some regions, leaving others
untouched. They therefore notice that the regions not affected by
hurricanes can be used as a counterfactual for the ones that are impacted.
They make the case that, by considering the forecast errors, they
are capable of eliminating the non-immediate effect of the disaster.
For example, if a region is particularly affected, the autoregressive
coefficients will capture this shift, and eliminate the long term
impact of hurricanes.

For this reason, if the units can contemporaneously be extracted to
be treated/controls, the burden of generating a counterfactual will
be both on the forecasts and on the other untreated units. Consider
the following structure for the policy:
\begin{assumption}
(Policy is dummy across heterogeneous units): \label{assu:(Policy-is-dummy}
Units can be divided into four subgroups, respectively 
\[
\begin{cases}
\widetilde{W}_{k,it}=1 & \text{ for }i\in I_{P}\text{ at times }t\in T_{P}\\
\widetilde{W}_{k,it}=0 & \text{ for }i\in I_{C}\text{ at times }t\in T_{P}\\
\widetilde{W}_{k,it}=0 & \text{ for }i\in I_{P}\text{ at times }t\in T_{C}\\
\widetilde{W}_{k,it}=0 & \text{ for }i\in I_{C}\text{ at times }t\in T_{C}
\end{cases}
\]
\end{assumption}
Here $I_{P}$ represents a subgroup of units which are treated, $I_{C}$
represents a subgroup of units which are not treated, $T_{P}$ represents
the subgroup of times that are treated, and $T_{C}$ represents the
subgroup of times that are not treated.\footnote{The $P$ notation is used to indicate ``policy'' times or ``policy''
affected units.} It follows that $I_{P}+I_{C}=I$ and $T_{P}+T_{C}=T$.

A consequence of assumption \ref{assu:(Policy-is-dummy} is that there
exist at least a control unit for times in which some units are subject
to a treatment. It is necessary in order to guarantee that $\widetilde{W}_{k,it}$
is a dummy that also includes non treated units that can act as contemporaneous
counterfactuals for the treated group.\footnote{Notice that this assumption also relates deeply to the results of
\citet{callawaysantannamultipletime2021}, which argue against fixed
effects estimation specifically in the case in which there are no
more control units after a series of treatments in the staggered treatment
framework. Assumption \ref{assu:(Policy-is-dummy} avoids such instances
completely by assuming that there are sufficient counterfactuals contemporaneously.}
\begin{assumption}
\label{assu:(Parallel-trends):}(Parallel trends): For each $j\geq1$,
$k\geq1$,$t\geq1$,$i\geq1$:
\[
\mathbb{E}[\widetilde{Y}_{j,it}(0)|t\in T_{P},i\in I_{P}]=\mathbb{E}[\widetilde{Y}_{j,it}(0)|t\in T_{P},i\in I_{C}].
\]
\end{assumption}
\begin{assumption}
\label{assu:(No-anticipation)}(No anticipation): For each $j\geq1$,
$k\geq1$,$t\geq1$$i\geq1$:
\[
\mathbb{E}[\widetilde{Y}_{j,it}(0)|t\in T_{C},i\in I_{P}]=\mathbb{E}[\widetilde{Y}_{j,it}(0)|t\in T_{C},i\in I_{C}].
\]
\end{assumption}
\begin{assumption}
\label{assu:(No-residual-autocorrelation):}(No residual autocorrelation):
For each $j\ge1$, $k\geq1$, $t\geq1$,,$i\geq1$ and for $s\geq1$: 
\[
corr(\widetilde{x}_{j,t},\widetilde{x}_{1,t-s}),..,corr(\widetilde{x}_{j,t},\widetilde{x}_{J,t-s})=0.
\]
\end{assumption}
\begin{assumption}
\label{assu:(No-contamination):}(No contamination):
For each $k\geq1$, $t\geq1$,$i\geq1$: 
\[
\widetilde{W}_{k,it} \perp \widetilde{W}_{-k,it}
\]
\end{assumption}

Assumption \ref{assu:(Parallel-trends):} implies that the treated
and non treated units are comparable in potential outcomes in treated
times. This is akin to a parallel trend assumption in the DiD literature,
with the difference that it is imposed in the forecast errors. Assumption
\ref{assu:(No-anticipation)} means that units belonging to the treatment
and control group have identical potential outcomes in non treated
times. Both of those assumptions are standard in the DiD literature.
Notice that an implication of the two assumptions is that they may
be violated in the case of spillover effects.

Assumption \ref{assu:(No-residual-autocorrelation):} states that
there must be no residual autocorrelation in the innovations of the
outcome variables. It is required to avoid the contamination of previous
causal effects or anticipatory effects. This is deeply rooted with
the idea of stationarity and its reliability strictly depends on the
nature of the data and the appropriate additions of dummies to control
for time series confounding factors in the PVAR regression, lag selection,
and the overall nature of the estimated system. This is because the
idea of this type of identification procedure is that any unit, at
any time, is allowed to be treated. 

Finally,  assumption \ref{assu:(No-contamination):} assumes independence across multiple treatments.  For example, if the economist is analysing the effects of two types of natural disasters, such as heatwaves and droughts, the assumption implies no sytematic correlation across the two assignments. \footnote{Notice that this assumption does still allow for assignments to be temporally related.  Consider the example of \citet{usman2024going}, which eliminate regions affected by multiple, subsequent, heterogeneous shocks, such as heatwaves and droughts, from the control sample.  From the point of view of a panel VAR, as long as the disasters are not \textit{contemporaneously} correlated, the occurrence of droughts after heatwaves will not contaminate inference as long as it is estimated by the AR coefficients of the two shocks.  Rather, the systematic occurrence of disasters at the same time could indeed contaminate inference, as there would be no way to distinguish the two dummy events.}

To understand the difference with respect to traditional causal inference
settings in panel data, such as DiD, staggered-DiD and LP-DiD, consider
Table \ref{tab:Treatment-scenarios-under}. 

\begin{table}
\centering{}{\small{}}%
\begin{tabular}{c|cc|ccc|ccc|cc}
 & \multicolumn{2}{c|}{{\footnotesize{}DiD}} & \multicolumn{3}{c|}{{\footnotesize{}Staggered DiD}} & \multicolumn{3}{c|}{{\footnotesize{}LP-DiD}} & \multicolumn{2}{c}{{\footnotesize{}PVAR Residuals}}\tabularnewline
\cline{2-11} \cline{3-11} \cline{4-11} \cline{5-11} \cline{6-11} \cline{7-11} \cline{8-11} \cline{9-11} \cline{10-11} \cline{11-11} 
{\footnotesize{}Units} & {\footnotesize{}Group 1} & {\footnotesize{}Group 2} & {\footnotesize{}Group 1} & {\footnotesize{}Group 2} & {\footnotesize{}Group 3} & {\footnotesize{}Group 1} & {\footnotesize{}Group 2} & {\footnotesize{}Group 3} & {\footnotesize{}Group 1} & {\footnotesize{}Group 2}\tabularnewline
\hline 
{\footnotesize{}$t=1$} & {\footnotesize{}0} & {\footnotesize{}0} & {\footnotesize{}0} & {\footnotesize{}0} & {\footnotesize{}0} & {\footnotesize{}0} & {\footnotesize{}0} & {\footnotesize{}0} & {\footnotesize{}0} & {\footnotesize{}1}\tabularnewline
{\footnotesize{}$t=2$} & {\footnotesize{}0} & {\footnotesize{}0} & {\footnotesize{}1} & {\footnotesize{}0} & {\footnotesize{}0} & {\footnotesize{}1} & {\footnotesize{}0} & {\footnotesize{}0} & {\footnotesize{}1} & {\footnotesize{}0}\tabularnewline
{\footnotesize{}$t=3$} & {\footnotesize{}1} & {\footnotesize{}0} & {\footnotesize{}1} & {\footnotesize{}1} & {\footnotesize{}0} & {\footnotesize{}1} & {\footnotesize{}1} & {\footnotesize{}0} & {\footnotesize{}0} & {\footnotesize{}1}\tabularnewline
{\footnotesize{}$t=4$} & {\footnotesize{}1} & {\footnotesize{}0} & {\footnotesize{}1} & {\footnotesize{}1} & {\footnotesize{}1} & {\footnotesize{}0} & {\footnotesize{}1} & {\footnotesize{}1} & {\footnotesize{}1} & {\footnotesize{}0}\tabularnewline
\hline 
\end{tabular}\caption{Treatment scenarios under different settings.\label{tab:Treatment-scenarios-under}}
\end{table}

The first column represents the traditional DiD literature. The economist
could compute the difference before and after period $t=3$ between
group 1 and group 2.

However, if more units can be included, the DiD becomes a staggered
setup like the one in column 4,5, and 6. The DiD literature (see,
among others, \citealp{callawaysantannamultipletime2021,sunabraham2021})
often times solves the problem of potentially shifting effects in
the treated units in staggered treatment designs by carefully selecting
a subgroup of units for which the econometrician knows that they did
not receive a treatment yet. Such units are supposed to act as a credible
counterfactual in a non-parametric framework. In this case, at time
$t=2$ the units in the group 1 receive a treatment, and the others
can be used as a control. It is possible to identify the effect of
the treatment at time $t=2$ by simply comparing the difference of
the units before and after the treatment. The same quantity can be
computed using the difference of $t=3$ and $t=1$. However, the $t=4$
observations cannot be used for causal inference. 

Clearly, a drawback of this approach is that, for an increasing $t$,
less and less counterfactuals are available. This is because the researcher
is not allowed to include the units which previously received a temporary
treatment in the control group. This is the reason why such methods
often limit themselves to the case of staggered treatment and not
a fully sparse treatment. Sparse treatments could be used in scenarios
in which the treatment becomes obsolete after one period, and allows
to use units from any group as counterfactual for the following period.

\citealp{dube2023local} show that, in the case of sparse treatments,
it is possible to avoid the underlying assumption of the staggered
DiD setup by imposing the existence of a time after which the units
can go back to the control group. This case may be more apt for scenarios
in which the researcher is aware of a short lasting effect of the
intervention. 

PVARs take a different approach. PVARs do not make use of the variables
directly, but of their residuals. Therefore, they can realistically
make the assumption that even treated units can become controls after
just one period. This interpretation is engrained in the idea of looking
at deviations from the autoregressive component of the treatment. Effectively, as long as the carryover effect is known, it is possible to use units which are still receiving some form of treatment by using their residuals. While \citet{dube2023local} would exclude units which are still receiving a form of treatment, PVARs allows them to be on the sample. Clearly, if the distribution of the carryover effects estimated by the autoregressive coefficients $\frac{cov(Y_{it},W_{i,t-j})}{W_{i,t-j}}$, is ill-behaved, inference will be contaminated.
This may be more credible in cases in which we believe that the autoregressive
component is capable of absorbing the effect of the shock, i.e. the
treatment does not alter the structural equilibrium of the outcome
variables. In this case, it can be reasonable to assume the following:
\begin{enumerate}
\item There is a cyclical component which motivates the use of a PVAR,
\item at least one unit is either treated or control at each time $t$ (assumption
\ref{assu:(Policy-is-dummy}),
\item units that have received a treatment would have otherwise seen a similar
outcome to the ones that did not (assumption \ref{assu:(Parallel-trends):}),
\item the AR coefficients are capable of eliminating the effects of the
treatment (assumption \ref{assu:(No-residual-autocorrelation):}).
\end{enumerate}
\begin{thm}
(PVAR: Randomised, dummy, heterogeneous policy): \label{thm:(PVAR:-Randomised,-dummy,}
Under assumption \ref{assu:SUTVA}\ref{assu:(Policy-is-dummy}, \ref{assu:(No-contamination):} Panel Vector Autoregressions
capture the following estimand:

\[
\begin{aligned}\gamma_{jk}= & \mathbb{E}[\widetilde{Y}_{j,it}|t\in T_{P},i\in I_{P}]-\mathbb{E}[\widetilde{Y}_{j,it}|t\in T_{P},i\in I_{C}]\\
 & -\mathbb{E}[\widetilde{Y}_{j,it}|t\in T_{C},i\in I_{P}]+\mathbb{E}[\widetilde{Y}_{j,it}|t\in T_{C},i\in I_{C}].
\end{aligned}
\]
\end{thm}
Hence, the causal effects can be easily found to be the ones in the
following Theorem:
\begin{thm}
(PVARs identify ATTs): \label{th:(identification-of-dummy-disturbances}
Under assumptions \ref{assu:SUTVA};\ref{assu:(Policy-is-dummy}; \ref{assu:(Parallel-trends):};
\ref{assu:(No-anticipation)} Panel Vector Autoregressions identify,
for all $j$, all $k$, all $t$, and all $i$:
\[
\gamma_{jk}=\widetilde{\text{ATT }}_{j,k,it}=\mathbb{E}[\widetilde{Y}_{j,it}(1)-\widetilde{Y}_{j,it}(0)|\widetilde{W}_{k,it}=1].
\]
\end{thm}

Notice that it is often times the case that inference with aggregate
variables is conducted with methods that do not have direct causal
interpretation. Even if in this specific case the assumptions may
appear to be too restrictive, it should be noted that the advantage
of this approach is the ability to make truly causal claims. The applied
literature has been historically welcoming to this assumption by the
means of DiD, since it involves a direct counterfactual and the credibility
of identification can be evaluated on a case by case scenario. This
is not true when dealing with aggregate data, where often times inference
is conducted without appropriately discussing the required assumptions
to make causal claims in potential outcome frameworks. 
\begin{rem}
DiD a-là-\citealp{callawaysantannamultipletime2021} should be preferred
if the researcher is in a staggered treatment setup. LP a-là \citealp{dube2023local}
may be preferred in cases in which the researcher is confident about
a date in which the effect of the policy terminates. PVARs should
be preferred if the researcher is confident about the temporary nature
of the intervention in the autoregressive residuals.
\end{rem}
\begin{rem}
To prove the results I have used a Cholesky decomposition, but no
independence assumption was necessary to obtain the ATT.
\end{rem}

\section{Application using natural disasters\label{sec:Applications}}

Due to the potentially detrimental effects of climate change related
disasters, much of the scientific literature has focussed on the issue
of trying to quantify the effects of disasters on the overall economy,
obtaining mixed results. Often times the effects are found to be neutral
or positive in real activity (see \citet{Strobl2011,linder2013}).
In the nominal side of the economy, it appears as though there is
consensus of a positive effect of disasters on inflation. For example,
\citet{beirne2022} finds a marked positive effect of natural disasters
on inflation in the euro area, \citet{parker2018} finds that natural
disasters positively affect inflation disproportionally in advanced
economies, and \citet{faccia2021feeling} find that extreme temperature
events have a non-negligible impact on prices. Regarding the timing,
there seems to be an agreement about the overall long term neutrality
of natural disasters (\citet{Cavallo2013}) but some empirical evidence
suggests a medium term effect of extreme climate events (\citet{usman2024going}).

The application I propose is the one that is at the most granular
level. I use the series on economic activity produced by \citet{baumeister2024tracking}
that tracks weekly economic activity at the state level in the United
States. The series is produced by considering data from various sources,
including data on mobility, labour market conditions, real activity
indicators, expectations, financial returns and household spending
and salaries. Moreover, I consider data on the natural disasters coming
from the NOAA\footnote{https://www.ncei.noaa.gov/access/billions/events/US/1980-2024?disasters{[}{]}=all-disasters},
which collects state level data on natural disasters exceeding the
billion dollar mark. Such disasters are classified into droughts,
floods, freezes, severe storms, tropical cyclones, wildfires and winter
storms. They tend to vary in length and impact, and a more granular
analysis of their impact is certainly warranted but outside the focus
of this paper. 

The classic assumption when evaluating the effect of disaster shocks
is to assume them to be exogenous with respect to a wide variety of
macroeconomic indicators. However, it may be that economic activity
is correlated with periods of high economic growth. This may overestimate
the disrupting effects of disasters if such growth is positively correlated
with deforestation, environmental degradation, or other geographical
features that may increase the damaging impact of the disaster. Many
other arguments could be put forth in favour of a violation of the
exogeneity condition.

However, if the researcher is willing to believe that territories
affected by natural disasters would have otherwise seen a level of
changes in economic activity which are the same ones as the one of
the territories not affected by disasters, in accordance to the previous
discussion, an ATT, rather than an ATE, may still be estimated. In
this case, I consider the following model
\begin{equation}\label{eq:disaster-pvar}
Y_{it}=(d_{it}^{\prime},\Delta e_{it}^{\prime})^{\prime},
\end{equation}
where $d_{it}$ is an indicator that is equal to one when a state
$i$ faces a natural disaster in week $t$, and $\Delta e_{it}$ is
the first difference weekly variable produced by \citet{baumeister2024tracking}\footnote{I use the first difference because the economic indicator is non stationary.}.
The data analyzed starts in April 1987 and ends in October 2024, as
that is the range of the series of \citet{baumeister2024tracking}.

The underlying assumptions behind this modelling approach are the following:
\begin{enumerate}
\item The potential outcome under no treatment of innovations in economic
activity of the units for which $\widetilde{d}_{it}\sim1$ at time
$t$ can be credibly represented by the units for which $\widetilde{d}_{it}\sim0$,
in accordance to assumption \ref{assu:(Parallel-trends):},
\item The potential outcome under no treatment in the innovations in economic
activity of the units for which $\widetilde{d}_{it}\sim0$ at time
$t$ is common to all the units, in accordance to assumption \ref{assu:(No-anticipation)},
\item The system presents no residual autocorrelation, in accordance to
assumption \ref{assu:(No-residual-autocorrelation):}.
\end{enumerate}
The PVAR is estimated through GMM, with five lags as it strikes the
best balance according to a MBIC test (see Table \ref{tab:Tests-for-lag-US}).
While the other tests may suggest other lags to be optimal, I focus
on the case of five lags as it also strikes the best economic intuition:
using the five previous weeks corresponds to making use of the previous
month data. Moreover, I add a dummy variable to control for the COVID-19
pandemic period. 
\begin{table}
\centering{}%
\begin{tabular}{cccc}
\hline 
 & MBIC & MAIC & MQIC\tabularnewline
\hline 
$p=1$ & 1547.81 & 1755.48 & 1706.31\tabularnewline
$p=2$ & 759.71 & 949.44 & 891.79\tabularnewline
$p=3$ & 228.60 & 380.39 & 334.27\tabularnewline
$p=4$ & 4.24 & 118.08 & 83.49\tabularnewline
$p=5$ & -34.68 & 41.21 & 18.16\tabularnewline
$p=6$ & -26.95 & 11.00 & -0.53\tabularnewline
\hline 
\hline 
 &  &  & \tabularnewline
\end{tabular}\caption{Tests for lag selection.\label{tab:Tests-for-lag-US}}
\end{table}
 
\begin{figure}
\begin{centering}
\includegraphics[scale=0.3]{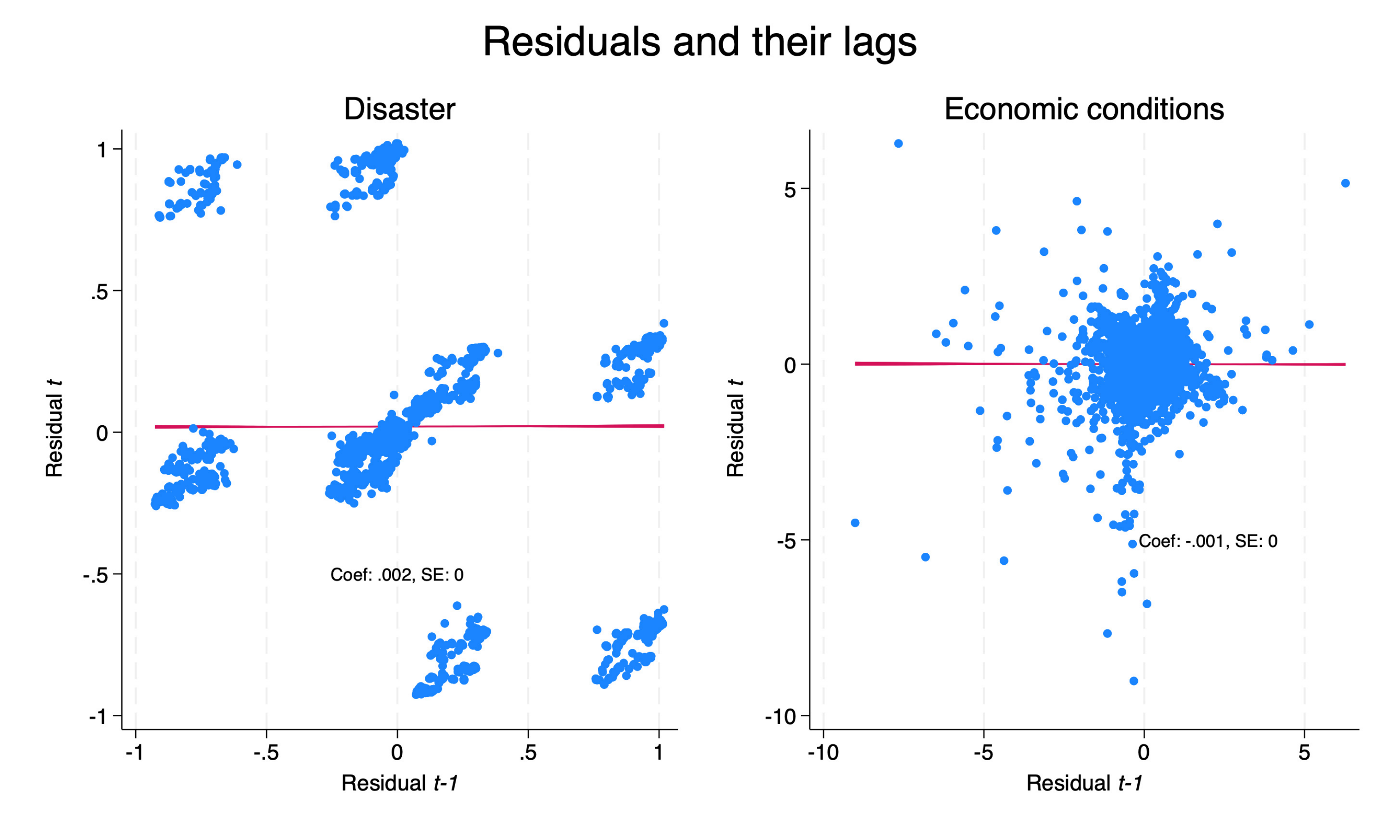}
\par\end{centering}
\centering{}\caption{Residuals of natural disasters and economic conditions plotted against
their lags.\label{fig:Residuals-of-econ-activity,US}}
\end{figure}
To verify that the no residual autocorrelation assumption holds, figure
\ref{fig:Residuals-of-econ-activity,US} displays the residuals of
each variable plotted against their lags. From the figure and the
regressions of the residuals against their lags, no significant correlation
emerges. 

The resulting Impulse Response Functions are plotted in figure \ref{fig:Impulse-Response-Functions-US},
which indicates a negative effect on economic activity that lasts
for two periods. The impulse response function suggests a negative
cumulative impact of $\sim.04$ on average coming from the fact that
the two weeks after the disaster appear to face a negative $\sim0.19$
deviation from their previous value.

\begin{figure}
\centering{}\includegraphics[scale=0.7]{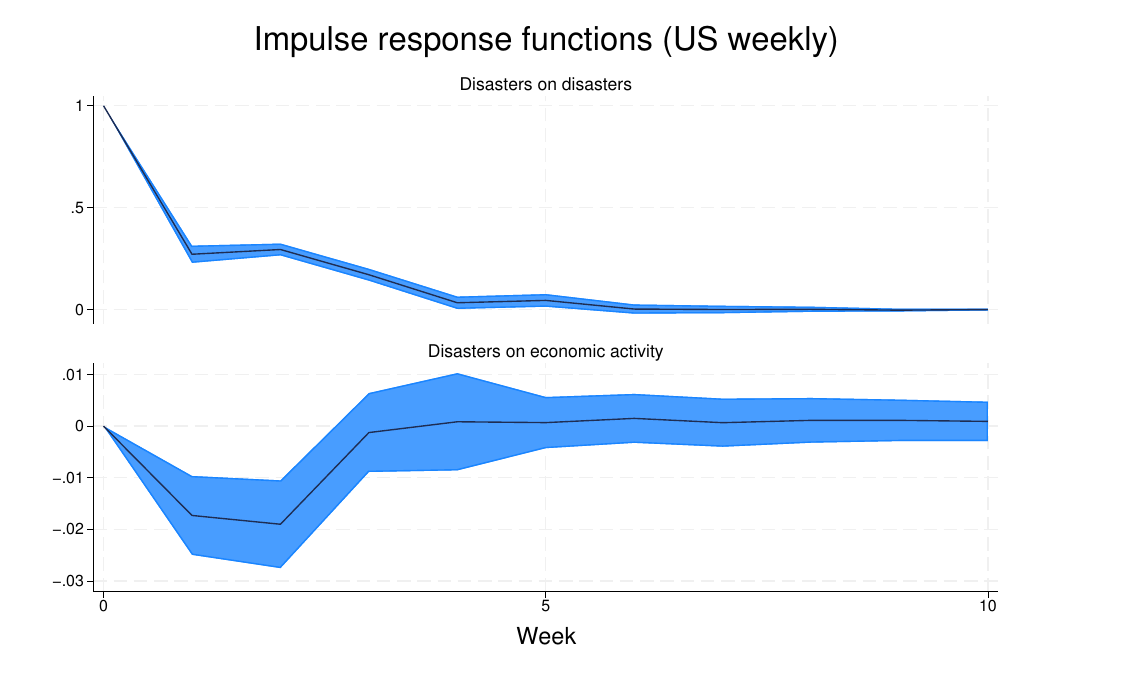}\caption{Impact of natural disasters on future natural disasters and economic
activity. Confidence intervals are bootstrapped with 1000 replications
at the $90\%$ level.\label{fig:Impulse-Response-Functions-US}}
\end{figure}

\section{Interference in panel vector autoregressions\label{sec:interference}}
According to the stable unit treatment value assumption (SUTVA \ref{assu:SUTVA}), potential outcomes only depend on one's own treatment assignment. In many cases, SUTVA may fail due to an unknown interference structure among neighbours. In the fields of environmental economics, urban economics, labour economics, and in the case of natural disasters, place-based policies and interventions often generate spillover effects. In the case of natural disasters, \citet{Cavallo2013} uses a synthetic control approach a-la-\citet{Abadie2010} to generate a counterfactual for countries that faced a significant natural disaster. They discuss possible SUTVA violations, and come to the conclusion that they are unlikely to play a significant role in the case of disasters in the medium term. Yet, \citet{bacchiocchi2024macroeconomic} find statistically significant short-run effects of severe weather shocks on local economic activity and cross-border spillovers operating through economic linkages between US states. Therefore, it seems reasonable to revise this paper's approach to inference and discuss potential violations of the SUTVA assumption and possible solutions.\\

Typically, the approach to dealing with interference through spillover effects is to explicitly assume a network structure (see, among others \citet{aronow2017estimating,manski2013identification,vazquez2023causal}). I follow their lead in explicitly modifying the potential outcome framework to the following
\begin{assumption}[Exposure mapping]\label{assu: exposure-mapping-spillover}
Define an exposure mapping $\mathcal{S}_{it}=s(\widetilde{W}_{-i,1:T})$ that summarizes the influence of other units' assignment paths on unit $i$. The potential outcome can be written as
\[
\widetilde{Y}_{it}\big(\widetilde{w}\;;\;\widetilde{W}_{-i,1:T}\big)
= \widetilde{Y}_{it}\big(\widetilde{W}_{i,1:t-1},\widetilde{w},\widetilde{W}_{i,t+1:T}\;;\;\mathcal{S}_{it}\big),
\]
\end{assumption}
Notice that assumption \ref{assu: exposure-mapping-spillover} allows potential spillovers among individual states due to a mapping $S_{it}$ which depends on the $-i$ units. For example, this may mean that if a natural disaster hits the $-i$ state of California, the effect may reach out to $i$ Arizona. 
In this context, several important causal effects can be discussed:
\begin{enumerate}
\item The total average treatment effect on the treated $\widetilde{\text{ATTE}}_{j,k,it} = \mathbb{E}[\widetilde{Y}_{j,it}(1,\mathcal{S}_{it}=s)-\widetilde{Y}_{j,it}(0,0)|\widetilde{W}_{k,it}=1,S_{j,j,it}=s]$
\item The average spillover effect on the treated  $\widetilde{\text{ASTE}}_{j,k,it} = \mathbb{E}[\widetilde{Y}_{j,it}(0,\mathcal{S}_{it}=s)-\widetilde{Y}_{j,it}(0,0)|\widetilde{W}_{k,it}=1,S_{j,j,it}=s]$
\end{enumerate}
Here,  the total average treatment effect on the treated represents the average effect of receiving a treatment and being subject to a form of spillover effect. The average spillover effect on the treated, instead, captures the effect of the spillover effect under a no treatment scenario.
One useful way to frame the estimation of causal effects is to identify what $\gamma_{j,k}$ could identify under the assumptions in section \ref{subsec:policy-heterog}:
\begin{thm}
(PVARs identify total effects and spillovers on the treated): \label{th:spillover-fail}
Under assumptions \ref{assu:(Policy-is-dummy}; \ref{assu:(Parallel-trends):};
\ref{assu:(No-anticipation)};\ref{assu: exposure-mapping-spillover} Panel Vector Autoregressions identify,
for all $j$, all $k$, all $t$, and all $i$:
\[
\gamma_{jk}=\widetilde{\text{ATTET}}_{j,k,it}-\widetilde{\text{ASTE}}_{j,k,it}
\]
\end{thm}

Essentially,  if the SUTVA assumption is violated, the researcher would claim the identification of an ATE, but would be capturing the spillover effect too.
This theorem resembles similar findings in different streams of literature. For example, \citet{vazquez2023identification} finds that experimental designs with spillover effects do not correctly estimate an average treatment effect. \citet{forastiere2021identification} finds that an estimator that wrongly assumes away interference is biased. \citet{butts2021difference,xu2023difference} find that DiD designs do not correctly identify an ATT in the case of spillovers. In essence, the theorem shows that assumptions  \ref{assu:(Policy-is-dummy}; \ref{assu:(Parallel-trends):} are not sufficient to capture a meaningful causal estimand in the context of interference. Rather, $\gamma_{j,k}$ isolates the average total effect on the treated - the impact of the treatment and the spillover - and the average spillover effect on the treated - the impact of spillovers on the counterfactual.

Several papers, however, propose different possible solutions to the issue of spillovers. If one can reasonably restrict the exposure mapping to specific neighbors, it becomes possible to estimate the spillover effect and therefore eliminate the source of bias.

For example, under a correctly specified exposure mapping $ S_{jk,it} $, the set of regressions
\begin{equation}\label{eq: estimation-spillover}
\widetilde{Y}_{j,it} = \delta_{jk} \,\widetilde{W}_{k,it} + \rho_{jk} \, S_{jk,it} + \eta_{jk,it},
\qquad \forall \, j,k,
\end{equation}
can estimate the total average treatment effect (on the treated) and the average spillover effect (on the treated) through $\delta$ and $\rho$, respectively. Typically, $ S_{jk,it} $ is assumed to be a function of neighbour distance. For example \citet{butts2021difference} revises \citet{kline2014local} using concentric county rings to measure the spillover exposure in the context of DiD. Another approach is proposed by \citet{xu2023difference}, which uses a propensity score matching (doubly robust estimator) to measure the spillover exposure.\footnote{\citet{xu2023difference} in particular does not specify an exact exposure map, but relies on the propensity score match. While this approach is not vulnerable to a misspecification of the exposure map, it does require a robust propensity score matching, whose accuracy depends on the information utilised to generate the matching function. It should also be noted that \citet{xu2023difference} operates in a finite population framework, which should provide less conservative inferential claims than superpopulation approaches. }

It is possible to recast the assumptions-estimands as follows. If the empirical application suits it, units which are immediately close to the treated ones may be affected by spillover effects, but units far away may still act as credible counterfactuals. For example, while Arizona may be subject to spillover effects if a drought hits California, it seems unrealistic to think that the effect may spill over to other states which are further away.\footnote{Especially in the context of the empirical application of this paper, usually NOAA well classifies the natural disasters and their economic impact. If a drought happens in California, but also negatively affects Arizona, the NOAA will classify the disaster as happening in Arizona too.} The most commonly utilised approach is to assume a known spillover matrix \footnote{For example, \citet{vazquez2023causal} formally assumes a distribution of compliance types in an instrumental variable framework.}

In this case, the following assumptions need to be satisfied in place of the common parallel trend and no anticipation.
\begin{assumption}\label{assu:modified-parallel-trends}(Modified parallel trends) For each $j\geq1$, $k\geq1$,$t\geq1$:
\[
\mathbb{E}[\widetilde{Y}_{jk}(0,s)|t\in T_{p},i\in I_{c};S_{j,k,it}=s]=E[Y(0,0)|t\in T_{p},i\in I_{p};S_{j,k,it}=0]
\]
\end{assumption}
\begin{assumption}\label{assu:modified-no-ancicip}(No lagged spillovers) For each $j\geq1$, $k\geq1$,$t\geq1$:
\[
\mathbb{E}[\widetilde{Y}_{jk,it}(0,0)|t\in T_{c},i\in I_{p};S_{j,k,it}=0]=\mathbb{E}[\widetilde{Y}_{j,k,it}(0,0)|t\in T_{c},i\in I_{c};S_{j,k,it}=0]
\]
\end{assumption}
Assumption \ref{assu:modified-parallel-trends} implies that the treated units would have seen similar economic developments to "far away" controls. Assumption \ref{assu:modified-no-ancicip} is only mildly stronger than \ref{assu:(No-anticipation)}. It reinforces the non-anticipation assumption to also require there not to be lagged spillover effects.

Under those assumptions, it is possible to show that the estimator $\delta_{j,k}$ in equation \ref{eq: estimation-spillover} recovers the average total effect on the treated.
\begin{thm}
(PVARs identify ATTs-spillovers): \label{th:spillover-identify}
Under assumptions \ref{assu:(Policy-is-dummy}; \ref{assu:modified-parallel-trends};
\ref{assu:modified-no-ancicip};\ref{assu: exposure-mapping-spillover} equation \ref{eq: estimation-spillover} identifies,
for all $j$, all $k$, all $t$, and all $i$:
\[
\delta_{j,k}=\widetilde{\text{ATTET}}_{j,k,it}
\]
\end{thm}

Once the $ATTE_{j,k,it}$ has been estimated, it can be plugged in as a value for the impulse response function \footnote{Notice that this approach, differently from the Cholesky decomposition, does not fully identify the rotations of the covariance matrix, $\hat{O}\hat{O}^(-1)=\hat{\Sigma}$, but it does provide enough information to generate the impulse response functions of interest.  In the equation for $\hat{O}$ in lemma \ref{lemmaa1} this means that, for a unitary shock in natural disasters $\hat{o}_{it}^{11}=1$, $\hat{o}_{it}^{12}=\delta$, but $\hat{o}^{22}_{it}$ is not identified. This feature is common in externally identified SVARS, such as \citet{gertler2015monetary}.}

Under this framework, it is possible to recover the errors from the PVAR estimation in \ref{eq:disaster-pvar} and estimate a linear regression against the natural disasters innovation and a spillover matrix. In this context, the spillover matrix is zero every period in which there are no disasters, and one in times in which there are disasters for the states that border the treated.\footnote{The matrix is generated based on USSWM from \citet{merryman2005}.} In this case $S_{it}$ becomes the inverse of how many states, on average, are treated in the neighbour. 

The results of the regression are displayed in table \ref{tab:spillover-regression-results}. They show an ATTE broadly similar to the ATT estimated by the Cholesky decomposition ($-.0109$), which does not take into account potential SUTVA violations.
\begin{table}[htbp]
\centering
\begin{tabular}{l c}
\toprule
 & $\widetilde{\text{Economic Activity}}$ \\
\midrule
$\widetilde{\text{Disasters}}$  & $-0.0140^{***}$ \\
                                & $(0.00462)$     \\
$\text{Spillovers}$             & $0.00548^{***}$ \\
                                & $(0.00180)$     \\
\bottomrule
\end{tabular}
\caption{Results of the linear regression of economic activity innovations
on natural disaster innovations and the spillover metric. 
Bootstrapped standard errors are reported in parentheses. 
Stars denote statistical significance at the 10\% (*), 5\% (**), and 1\% (***) levels.}\label{tab:spillover-regression-results}
\end{table}
The estimated impulse response is very similar to the original work because its evolution strictly depends on the AR coefficients, which are the same as the specification in section \ref{sec:Applications}, and is therefore not presented in this text to avoid redundancies.

\section{Conclusion \label{sec:Conclusion}}

Panel Vector Autoregressions are a flexible tool for the estimation
of causal effects. Depending on the shape of the treatment and the
assumptions the researcher is willing to put on their distributions,
the causal effects can either be:
\begin{enumerate}
\item Average Treatment Effects (ATE) $\mathbb{E}[\widetilde{Y}_{j,it}(1)-\widetilde{Y}_{j,it}(0)]$
if the policy variable is an ``homogeneous'' dummy, i.e. all the
units are either jointly treated or jointly non treated at each time,
and is independent on its past, its future, other policy variables,
and the outcome variable.
\item Average Causal Response on the Treated (ACRT) $\frac{\delta\mathbb{E}[\widetilde{Y}_{j,it}(\lambda_{k})|\widetilde{W}_{k,it}=\lambda_{k}]}{\delta\lambda_{k}}$
if the policy variable is heterogeneous among units, continuous, and
normally distributed.
\item A weighted average of Average Causal Response on the Treated (ACRT)
and Average Treatment Effects (ATE) $\int_{d_{L}}^{d_{U}}\frac{\delta\mathbb{E}[\widetilde{Y}_{j,it}(\lambda_{k})|\widetilde{W}_{kit}=\widetilde{w}_{k}^{\prime}]}{\delta\lambda_{k}}q_{1}(\lambda_{k})+q_{0}\frac{\mathbb{E}[\widetilde{Y}_{j,it}(d_{L})-\widetilde{Y}_{j,it}(0)]}{d_{L}}$
if the policy variable is heterogeneous among units in a given time,
non-negative, continuous, and strong parallel trends hold.
\item Average Treatment Effects on the Treated (ATT) $\mathbb{E}[\widetilde{Y}_{j,it}|\widetilde{W}_{k,it}=1]-\mathbb{E}[\widetilde{Y}_{j,it}|\widetilde{W}_{k,it}=0]$
if the policy variable is a dummy but some units are treated while
others are not, the units respect the ``parallel trend'' and ``no
anticipation'' conditions, and there is no residual autocorrelation.
\end{enumerate}
I showcase the high potential of this last set of assumptions
to evaluate the effects of natural disasters on the real economy. To best complement the work, I provide alternative formulations to deal with spillovers.

\newpage
\bibliographystyle{apalike}
\bibliography{bibliography}


\newpage
\appendix

\section*{Mathematical results\label{sec:AppendixACPVAR}}

\subsection{Proofs}
\begin{lem}
\label{lemmaa1}Consider (without loss of generality)\footnote{The element in the first column, $n-th$ row, is always $\frac{cov(\widetilde{W}_{it},\widetilde{Y}_{it}^{n-1})}{\sqrt{\widetilde{W}_{it}}}$
if there is only one policy variable, which allows to make the claim
that the result is valid for any outcome variable.} the case of a system of the kind $x_{it}=(W_{it}^{\prime},Y_{it}^{\prime})^{\prime}$,
which allows one outcome and one policy. The Cholesky decomposition
of the covariance matrix of the disturbances results in 
\[
\hat{O}=\left[\begin{array}{cc}
\sqrt{var(\widetilde{W}_{it})} & 0\\
\frac{cov(\widetilde{W}_{it},\widetilde{Y}_{it})}{\sqrt{var(\widetilde{W}_{it})}} & \sqrt{var(\widetilde{Y}_{it})-(\frac{cov(\widetilde{W}_{it},\widetilde{Y}_{it})}{\sqrt{var(\widetilde{W}_{it})}})^{2}}
\end{array}\right]=\left[\begin{array}{cc}
\hat{o}_{it}^{11} & 0\\
\hat{o}_{it}^{12} & \hat{o}_{it}^{22}
\end{array}\right].
\]
Notice that normally the researcher is interested in a unitary response
to a shock in the first variable. This is done by considering the
impulse $\hat{O}*(1/\sqrt{var(\widetilde{W}_{it})},0)^{\prime}$.
Hence, the estimator of the effect of a shock in the first variable
to the second becomes $\hat{\gamma}=\hat{o}_{it}^{12}/\hat{o}_{it}^{11}$,
which is
\[
\frac{cov(\widetilde{W}_{it},\widetilde{Y}_{it})}{var(\widetilde{W}_{it})}
\]
Hence, there is a $\hat{\gamma}_{jk}$ for each $j=1,..,J$.

\end{lem}
\begin{proof}
\textbf{Proof of Theorem \ref{thm:(PVARs-Impulse-Response}.} 
Consider
the first part of the Impulse Response Function, i,e, $\mathbb{E}[\widetilde{Y}_{j,it}1\{\widetilde{W}_{k,it}=1\}]$.

\[
\begin{aligned}\mathbb{E}[\widetilde{Y}_{j,it}1\{\widetilde{W}_{k,it}=1\}]\\
=\mathbb{E}[\widetilde{Y}_{j,it}(\widetilde{W}_{i,1:t-1},1,\widetilde{W}_{-k,it},\widetilde{W}_{i,t+1:t+h})1\{\widetilde{W}_{k,it}=1\}]\\
=\mathbb{E}[\widetilde{Y}_{j,it}(\widetilde{W}_{i,1:t-1},1,\widetilde{W}_{-k,it},\widetilde{W}_{i,t+1:t+h})]\mathbb{E}[1\{\widetilde{W}_{k,it}=1\}]\\
+Cov(\widetilde{Y}_{j,it}(\widetilde{W}_{i,1:t-1},1,\widetilde{W}_{-k,i,t},\widetilde{W}_{i,t+1:t+h}),1\{\widetilde{W}_{k,it}=1\})
\end{aligned}
\]

Where the first equality comes from writing the first part of the
IRF, the second comes from expressing the realized outcomes in terms
of the potential outcomes, the third comes from the rule of the expectations
$\mathbb{E}[AB]=\mathbb{E}[A]\mathbb{E}[B]+cov(A,B)$.

Hence:
\[
\begin{aligned}\mathbb{E}[\widetilde{Y}_{j,it}|\widetilde{W}_{k,it}=1]=\mathbb{E}[\widetilde{Y}_{j,it}(\widetilde{W}_{i,1:t-1},1,\widetilde{W}_{-k,it},\widetilde{W}_{i,t+1:t+h})]\\
+\frac{Cov(\widetilde{Y}_{j,it}(\widetilde{W}_{i,1:t-1},1,\widetilde{W}_{-k,it},\widetilde{W}_{i,t+1:t+h}),1\{\widetilde{W}_{k,it}=1\})}{\mathbb{E}[1\{\widetilde{W}_{k,it}=1\}]}
\end{aligned}
\]

Finally, considering the other side of the Impulse Response Function,
I obtain the desired result.\footnote{This claim is similar to the one of \citet{RambachanSheppard2021},
with the difference that here I am specifically discussing the result
of a Cholesky decomposition, which moves the interpretation to the
residuals.}
\end{proof}

\begin{proof}
\textbf{Proof of Theorem \ref{thm:(PVARs-estimate-ATEs).}. }The proof
is in the text.
\end{proof}
\begin{proof}
\textbf{Proof of Theorem \ref{thm:(PVARS:-Randomised,-continuous,}.
}Rewriting Lemma \ref{lemmaa1}, the proposed estimator estimates,
for every dependent variable $Y_{j,i}$ and every policy $W_{k,it}$
\[
\gamma_{jk}=\int q(\lambda_{k})g'(\lambda_{k})d\lambda_{k},
\]

Consider that the estimator is $Cov(\widetilde{Y}_{j},\widetilde{W}_{k})/var(\widetilde{W}_{k})$. 

The numerator will be:

\[
\begin{aligned}Cov(\widetilde{Y}_{j},\widetilde{W}_{k}) & =\mathbb{E}[(\widetilde{Y}_{j}-\mathbb{E}(\widetilde{Y}_{j}))(\widetilde{W}_{k}-\mathbb{E}(\widetilde{W}_{k}))]\\
 & =\mathbb{E}[\widetilde{Y}_{j}(\widetilde{W}_{k}-\mathbb{E}(\widetilde{W}_{k}))]\\
 & =\mathbb{E}[(\mathbb{E}[\widetilde{Y}_{j}|\widetilde{W}_{k}])(\widetilde{W}_{k}-\mathbb{E}(\widetilde{W}_{k}))]\\
 & =\int(\lambda_{k}-\mathbb{E}(\widetilde{W}_{k}))g(\lambda_{k})f_{\widetilde{W}_{k}}(\lambda_{k})d\lambda_{k},
\end{aligned}
\]

where the first equality holds because of the law of the covariance,
the second holds because the innovations of a VAR are assumed to be
zero mean for each $j$ and each $t$, the third equality holds by
the law of total expectations, and the last holds by rewriting the
expected value as an integral and defining $g(\lambda_{k})=\mathbb{E}[\widetilde{Y}_{j}|\widetilde{W}_{k}=\lambda_{k}]$.
Defining $v^{\prime}(m)=(\lambda_{k}-\mathbb{E}[\widetilde{W}_{k}])f_{\widetilde{W}}(m)$,
$v(m)=\int_{-\infty}^{\widetilde{W}_{k}}(m-\mathbb{E}[\widetilde{W}_{k}])f_{\widetilde{W}_{k}}(m)dm$
and $u(\lambda_{k})=g(\lambda_{k})$ I can apply integration by parts
to obtain 
\begin{equation}
\begin{aligned}Cov(\widetilde{Y}_{j},\widetilde{W}_{k}) & =\end{aligned}
\int_{-\infty}^{\widetilde{W}_{k}}(m-\mathbb{E}[\widetilde{W}_{k}])f_{\widetilde{W}_{k}}(m)dm\,g(\lambda_{k})-\int_{-\infty}^{\infty}(\int_{-\infty}^{\widetilde{W}_{k}}(m-\mathbb{E}[\widetilde{W}_{k}])f_{\widetilde{W}_{k}}(m)dm\,g^{\prime}(\lambda_{k})d\lambda_{k}.\label{eq:equazione_prova_random}
\end{equation}

Notice that the first part converges to zero if the variance of $\widetilde{W}_{k}$
exists, and changing the sign to the second part we obtain
\[
\begin{aligned}Cov(\widetilde{Y}_{j},\widetilde{W}_{k}) & =\int_{-\infty}^{\infty}(\int_{-\infty}^{\widetilde{W}_{k}}(\mathbb{E}[\widetilde{W}_{k}]-m)f_{\widetilde{W}_{k}}(m)dm\,g^{\prime}(\lambda_{k})d\lambda_{k}\\
 & =\int_{-\infty}^{\infty}(\mathbb{E}[\widetilde{W}_{k}]F_{\widetilde{W}_{k}}(\lambda_{k})-\theta_{\widetilde{W}_{k}}(\lambda_{k}))g^{\prime}(\lambda_{k})d\lambda_{k},
\end{aligned}
\]

where the first equality holds by changing the sign of the second
part of \ref{eq:equazione_prova_random}, the second holds by substituting
the definition of $\theta_{\widetilde{W}_{k}}(\lambda_{k})=\int_{-\infty}^{\widetilde{W}_{k}}mf_{\widetilde{W}_{k}}(m)dm$.

And the denominator is $var(\widetilde{W}_{k})=\sigma_{\widetilde{W}_{k}}^{2}$.
Therefore,
\[
\gamma_{jk}=\frac{\int_{-\infty}^{\infty}(\mathbb{E}[\widetilde{W}_{k}]F_{\widetilde{W}_{k}}(\lambda_{k})-\theta_{\widetilde{W}_{k}}(\lambda_{k}))g^{\prime}(\lambda_{k})d\lambda_{k}}{\sigma_{\widetilde{W}_{k}}^{2}},
\]

which is equivalent to the one in the theorem by using the definition
of the weights $q(\lambda_{k})=\frac{1}{\sigma_{\widetilde{W}_{k}}^{2}}\int_{-\infty}^{\infty}(\mathbb{E}[\widetilde{W}_{k}]F_{\widetilde{W}_{k}}(\lambda_{k})-\theta_{\widetilde{W}_{k}}(\lambda_{k}))$.
\end{proof}
\begin{proof}
\textbf{Proof of Theorem \ref{thm:(PVARs-estimate-the-ACRT)} }Consider
the form 
\[
\gamma_{jk}=\int q(\lambda_{k})g'(\lambda_{k})d\lambda_{k}.
\]

Substituting $q(\lambda_{k})=\frac{1}{\sqrt{2\pi}}\int_{-\infty}^{\lambda_{k}}me^{-m^{2}/2}dm=\frac{1}{\sqrt{2\pi}}e^{-m^{2}/2}$
inside the definition of $\gamma_{k}$, I obtain the following 
\[
\begin{aligned}\gamma_{jk} & =\int\frac{1}{\sqrt{2\pi}}e^{-m^{2}/2}g^{\prime}(\lambda_{k})d\lambda_{k}\\
 & =g^{\prime}(\lambda_{k})\int\frac{1}{\sqrt{2\pi}}e^{-m^{2}/2}d\lambda_{k}\\
 & =g^{\prime}(\lambda_{k}),
\end{aligned}
\]

where the first equality is true by substituting the value of the
weights $q(\lambda_{k})$, the second by the fact that the density
of $g^{\prime}(\lambda_{k})$ does not depend on $\lambda_{k}$, and
the last by the laws of integration of normal variables which are
satisfied by $\widetilde{W}_{k}$ according to assumption \ref{assu:(Normal-distribution)}.
\end{proof}
\begin{proof}
\textbf{Proof of Theorem \ref{thm:(PVARs-identify-ACR):}. }The ACRT
can be expressed as:
\[
\frac{\delta\mathbb{E}[\widetilde{Y}_{j}(\lambda_{k})|\widetilde{W}_{k}=\lambda_{k}]}{\delta\lambda_{k}}=\frac{\delta\mathbb{E}[\widetilde{Y}_{j}(\lambda_{k})]}{\delta\lambda_{k}}+\frac{cov(\widetilde{Y}_{j}(\lambda_{k}),1\{\widetilde{W}_{k}=\lambda_{k}\})}{var(1\{\widetilde{W}_{k}=\lambda_{k}\})}.
\]

By the same argument as Theorem \ref{thm:(PVARs-estimate-ATEs).}
$cov(\widetilde{Y}_{j}(\lambda_{k}),1\{\widetilde{W}_{k}=\lambda_{k}\})=0$,
which concludes the proof.
\end{proof}
\begin{proof}
\textbf{Proof of Theorem \ref{thm:(PVARs-estimate-weighted-average)}
}The proposed estimator is:
\[
\gamma_{jk}=cov(\widetilde{Y}_{j,it},\widetilde{W}_{k,it})/var(\widetilde{W}_{k,it})
\]

To shorten the definitions I will use $\ensuremath{m(d)=\mathbb{E}[\widetilde{Y}_{j,it}|\widetilde{W}_{k,it}=d]}$.

\[
\begin{aligned}\gamma_{jk}= & \frac{\mathbb{E}[\widetilde{Y}_{j,it}(\widetilde{W}_{k,it}-\mathbb{E}[\widetilde{W}_{k,it}])]}{var(\widetilde{W}_{k,it})}\\
 & =\mathbb{E}[\frac{(\widetilde{W}_{k,it}-\mathbb{E}[\widetilde{W}_{k,it}])}{var(\widetilde{W}_{k,it})}(m(\widetilde{W}_{k,it})-m(0))]\\
 & =\mathbb{E}\left[\frac{(\widetilde{W}_{k,it}-\mathbb{E}[\widetilde{W}_{k,it}])}{var(\widetilde{W}_{k,it})}(m(\widetilde{W}_{k,it})-m(0))|\widetilde{W}_{k,it}>0\right]\mathbb{P}(\widetilde{W}_{k,it}>0)\\
 & =\mathbb{E}\left[\frac{(\widetilde{W}_{k,it}-\mathbb{E}[\widetilde{W}_{k,it}])}{var(\widetilde{W}_{k,it})}(m(\widetilde{W}_{k,it})-m(d_{L}))|\widetilde{W}_{k,it}>0\right]\mathbb{P}(\widetilde{W}_{k,it}>0)+\\
 & +\mathbb{E}\left[\frac{(\widetilde{W}_{k,it}-\mathbb{E}[\widetilde{W}_{k,it}])}{var(\widetilde{W}_{k,it})}(m(d_{L})-m(0))|\widetilde{W}_{k,it}>0\right]\mathbb{P}(\widetilde{W}_{k,it}>0)=\\
 & =A_{1}+A_{2}.
\end{aligned}
\]

Here the first equality holds because the proposed estimator is akin
to a simple linear regression of $\widetilde{Y}_{j,it}$ on $\widetilde{W}_{k,it}$,
the second equality holds because $\mathbb{E}[(\widetilde{W}_{k,it}-\mathbb{E}[\widetilde{W}_{k,it}])m(0)]=0$,
the third equality holds because $\mathbb{E}[m(\widetilde{W}_{k,it})-m(0)|\widetilde{W}_{k,it}=0]=0$,
and the fourth equality holds by adding and subtracting $m(d_{L})$.
Then, for $A_{1}$:
\[
\begin{aligned}A_{1}= & \mathbb{E}\left[\frac{(\widetilde{W}_{k,it}-\mathbb{E}[\widetilde{W}_{k,it}])}{var(\widetilde{W}_{k,it})}(m(\widetilde{W}_{k,it})-m(d_{L}))|\widetilde{W}_{k,it}>0\right]\mathbb{P}(\widetilde{W}_{k,it}>0)\\
 & =\frac{\mathbb{P}(\widetilde{W}_{k,it}>0)}{var(\widetilde{W}_{k,it})}\int_{d_{L}}^{d_{U}}(c-\mathbb{E}[\widetilde{W}_{k,it}])(m(c)-m(d_{L}))dF_{\widetilde{W}_{k,it}|\widetilde{W}_{k,it}>0}(c)\\
 & =\frac{\mathbb{P}(\widetilde{W}_{k,it}>0)}{var(\widetilde{W}_{k,it})}\int_{d_{L}}^{d_{U}}(c-\mathbb{E}[\widetilde{W}_{k,it}])\int_{d_{L}}^{k}m'(\lambda)d\lambda dF_{\widetilde{W}_{k,it}|\widetilde{W}_{k,it}>0}(c)\\
 & =\frac{\mathbb{P}(\widetilde{W}_{k,it}>0)}{var(\widetilde{W}_{k,it})}\int_{d_{L}}^{d_{U}}(c-\mathbb{E}[\widetilde{W}_{k,it}])\int_{d_{L}}^{d_{U}}\boldsymbol{1}\{\lambda<c\}m'(\lambda)d\lambda dF_{\widetilde{W}_{k,it}|\widetilde{W}_{k,it}>0}(c)\\
 & =\frac{\mathbb{P}(\widetilde{W}_{k,it}>0)}{var(\widetilde{W}_{k,it})}\int_{d_{L}}^{d_{U}}m'(\lambda)\int_{d_{L}}^{d_{U}}(k-\mathbb{E}[\widetilde{W}_{k,it}])\boldsymbol{1}\{\lambda<c\}dF_{\widetilde{W}_{k,it}|\widetilde{W}_{k,it}>0}(c)d\lambda\\
 & =\frac{\mathbb{P}(\widetilde{W}_{k,it}>0)}{var(\widetilde{W}_{k,it})}\int_{d_{L}}^{d_{U}}m'(\lambda)\mathbb{E}[(\widetilde{W}_{k,it}-\mathbb{E}[\widetilde{W}_{k,it}])\boldsymbol{1}\{\lambda\leq\widetilde{W}_{k,it}\}|\widetilde{W}_{k,it}>0]d\lambda\\
 & =\frac{\mathbb{P}(\widetilde{W}_{k,it}>0)}{var(\widetilde{W}_{k,it})}\int_{d_{L}}^{d_{U}}m'(\lambda)\mathbb{E}[(\widetilde{W}_{k,it}-\mathbb{E}[\widetilde{W}_{k,it}])|\widetilde{W}_{k,it}\geq\lambda]\mathbb{P}[\widetilde{W}_{k,it}\geq\lambda|\widetilde{W}_{k,it}>0]d\lambda\\
 & =\int_{d_{L}}^{d_{U}}m'(\lambda)\frac{(\mathbb{E}[\widetilde{W}_{k,it}|\widetilde{W}_{k,it}\geq\lambda]-\mathbb{E}[\widetilde{W}_{k,it}])\mathbb{P}[\widetilde{W}_{k,it}\geq\lambda]}{var(\widetilde{W}_{k,it})}d\lambda,
\end{aligned}
\]

where the first equality rewrites the previous result, the second
holds by rearranging and writing the expectation as an integral, the
third makes use of the fundamental theorem of calculus, the fourth
rewrites the inner integral so that it is over $d_{U}$ to $d_{L}$,
the fifth changes the order of integration and rearranging terms,
the sixth by rewriting the inner integral as an expectation, the seventh
by the law of iterated expectations (and since $\widetilde{W}_{k,it}\geq\lambda$
implies that $\widetilde{W}_{k,it}\geq0$), and the last equality
holds by combining terms. For $A_{2}$: 
\[
\begin{aligned}A_{2} & =\mathbb{E}\left[\frac{(\widetilde{W}_{k,it}-\mathbb{E}[\widetilde{W}_{k,it}])}{var(\widetilde{W}_{k,it})}(m(d_{L})-m(0))|\widetilde{W}_{k,it}>0\right]\mathbb{P}[\widetilde{W}_{k,it}>0]\\
 & =\frac{(\mathbb{E}[\widetilde{W}_{k,it}|\widetilde{W}_{k,it}>0]-\mathbb{E}[\widetilde{W}_{k,it}])\mathbb{P}(\widetilde{W}_{k,it}>0)d_{L}}{var(\widetilde{W}_{k,it})}\frac{(m(d_{L})-m(0)}{d_{L}},
\end{aligned}
\]

where the first equality is the definition of $A_{2}$, and the second
holds by multiplying and dividing by $d_{L}$. Then we obtain:
\[
\begin{aligned}\gamma_{jk} & =\int_{d_{L}}^{d_{U}}m'(\lambda)\frac{(\mathbb{E}[\widetilde{W}_{k,it}|\widetilde{W}_{k,it}\geq\lambda]-\mathbb{E}[\widetilde{W}_{k,it}])\mathbb{P}[\widetilde{W}_{k,it}\geq\lambda]}{var(\widetilde{W}_{k,it})}d\lambda+\\
 & +\frac{(\mathbb{E}[\widetilde{W}_{k,it}|\widetilde{W}_{k,it}>0]-\mathbb{E}[\widetilde{W}_{k,it}])\mathbb{P}(\widetilde{W}_{k,it}>0)d_{L}}{var(\widetilde{W}_{k,it})}\frac{(m(d_{L})-m(0))}{d_{L}},
\end{aligned}
\]

which can be rewritten as: 
\[
\gamma_{jk}=\int_{d_{L}}^{d_{U}}q_{1}(\lambda)\frac{\delta\mathbb{E}[\widetilde{Y}_{j,it}|\widetilde{W}_{k,it}=\lambda]}{d\lambda}d\lambda+q_{0}\frac{\mathbb{E}[\widetilde{Y}_{j,it}|\widetilde{W}_{k,it}=d_{L}]-\mathbb{E}[\widetilde{Y}_{j,it}|\widetilde{W}_{k,it}=0]}{d_{L}},
\]
 where: 
\[
\begin{aligned}q_{1}(\lambda):=\frac{\mathbb{E}[\widetilde{W}_{k,it}|\widetilde{W}_{k,it}\geq\lambda]-\mathbb{E}[\widetilde{W}_{k,it}])\mathbb{P}(\widetilde{W}_{k,it}\geq\lambda)}{var(\widetilde{W}_{k,it})}, & \text{ and } & q_{0}:=\frac{(\mathbb{E}[\widetilde{W}_{k,it}|\widetilde{W}_{k,it}>0]-\mathbb{E}[\widetilde{W}_{k,it}])\mathbb{P}(\widetilde{W}_{k,it}>0)d_{L}}{var(\widetilde{W}_{k,it})}\end{aligned}
.
\]
 Here the weights satisfy: 
\[
\begin{aligned}\int_{d_{L}}^{d_{U}}q_{1}(\lambda)d\lambda+w_{0}= & \frac{1}{var(\widetilde{W}_{k,it})}\{\int_{d_{L}}^{d_{U}}\mathbb{E}[\widetilde{W}_{k,it}|\widetilde{W}_{k,it}\geq\lambda]\mathbb{P}[\widetilde{W}_{k,it}\geq\lambda)d\lambda\\
 & -\mathbb{E}[\widetilde{W}_{k,it}]\int_{d_{L}}^{d_{L}}\mathbb{P}(\widetilde{W}_{k,it}\geq\lambda)d\lambda\\
 & +\mathbb{E}[\widetilde{W}_{k,it}|\widetilde{W}_{k,it}>0]\mathbb{P}(\widetilde{W}_{k,it}>0)d_{L}\\
 & -\mathbb{E}[\widetilde{W}_{k,it}]\mathbb{P}(\widetilde{W}_{k,it}>0)d_{L}\}\\
 & :=\frac{1}{var(\widetilde{W}_{k,it})}\{B_{1}+B_{2}+B_{3}+B_{4}\},
\end{aligned}
\]
 where for $B_{1}$, for all $\lambda\in\mathcal{W}_{+}$: 
\[
\begin{aligned}\mathbb{E}[\widetilde{W}_{k,it}|\widetilde{W}_{k,it}\geq\lambda]\mathbb{P}(\widetilde{W}_{k,it}\geq\lambda) & =\mathbb{E}[\widetilde{W}_{k,it}\boldsymbol{1}\{\widetilde{W}_{k,it}>\lambda\}|\widetilde{W}_{k,it}\geq\lambda]\mathbb{P}[\widetilde{W}_{k,it}\geq\lambda]\\
 & =\mathbb{E}[\widetilde{W}_{k,it}\boldsymbol{1}\{\widetilde{W}_{k,it}\geq\lambda\}],
\end{aligned}
\]
 which holds by the law of iterated expectations and implies that
\[
\begin{aligned}B_{1} & =\int_{d_{L}}^{d_{U}}\mathbb{E}[\widetilde{W}_{k,it}|\widetilde{W}_{k,it}\geq\lambda]\mathbb{P}(\widetilde{W}_{k,it}\geq\lambda)d\lambda\\
 & \int_{d_{L}}^{d_{U}}\int_{\mathcal{W}}d\boldsymbol{1}\{d\geq\lambda\}dF_{\mathcal{W}}(d)d\lambda\\
 & \int_{\mathcal{W}}d\left(\int_{d_{L}}^{d_{U}}\boldsymbol{1}\{\lambda\leq d\}d\lambda\right)dF_{\mathcal{W}}(d)\\
 & =\int_{\mathcal{W}}d(d-d_{L})dF_{\mathcal{W}}(d)\\
 & =\mathbb{E}[\widetilde{W}_{k,it}^{2}]-\mathbb{E}[\widetilde{W}_{k,it}]d_{L},
\end{aligned}
\]
 where the first line is $B_{1}$, the second holds by the previous
result, the third by changing the order of integration, the fourth
by carrying out the inner integration, and the last by rewriting the
integral as an expectation. Next, for $B_{2}$:
\[
\begin{aligned}B_{2} & =\mathbb{E}[\widetilde{W}_{k,it}]\int_{d_{L}}^{d_{U}}\mathbb{P}(\widetilde{W}_{k,it}\geq\lambda)d\lambda\\
 & =\mathbb{E}[\widetilde{W}_{k,it}]\mathbb{P}(\widetilde{W}_{k,it}>0)\int_{d_{L}}^{d_{U}}\mathbb{P}(\widetilde{W}_{k,it}\geq\lambda|\widetilde{W}_{k,it}>0)d\lambda\\
 & =\mathbb{E}[\widetilde{W}_{k,it}]\mathbb{P}(\widetilde{W}_{k,it}>0)\int_{d_{L}}^{d_{U}}\int_{d_{L}}^{d_{U}}\boldsymbol{1}\{d\leq\lambda\}dF_{\widetilde{W}_{k,it}|\widetilde{W}_{k,it}>0}(d)d\lambda\\
 & =\mathbb{E}[\widetilde{W}_{k,it}]\mathbb{P}(\widetilde{W}_{k,it}>0)\int_{d_{L}}^{d_{U}}\left(\int_{d_{L}}^{d_{U}}\boldsymbol{1}\{d\leq\lambda\}d\lambda\right)dF_{\widetilde{W}_{k,it}|\widetilde{W}_{k,it}>0}(d)\\
 & =\mathbb{E}[\widetilde{W}_{k,it}]\mathbb{P}(\widetilde{W}_{k,it}>0)\int_{d_{L}}^{d_{U}}(d-d_{L})dF_{\widetilde{W}_{k,it}|\widetilde{W}_{k,it}>0}(d)\\
 & =\mathbb{E}[\widetilde{W}_{k,it}]\mathbb{P}(\widetilde{W}_{k,it}>0)(\mathbb{E}[\widetilde{W}_{k,it}|\widetilde{W}_{k,it}>0]-d_{L})\\
 & =\mathbb{E}[\widetilde{W}_{k,it}]^{2}-\mathbb{E}[\widetilde{W}_{k,it}]\mathbb{P}[\widetilde{W}_{k,it}>0]d_{L},
\end{aligned}
\]
 where the first equality is the definition of $B_{2}$, the second
equality holds by the law of iterated expectations, the third equality
holds by writing $\mathbb{P}(\widetilde{W}_{k,it}\geq\lambda|\widetilde{W}_{k,it}>0)$
as an integral, the fourth equality changes the order of integration,
the fifth equality carries out the inside integration, the sixth equality
rewrites the integral as an expectation, the last equality holds by
combining terms and by the law of iterated expectations. Then for
$B_{3}$: 
\[
\begin{aligned}B_{3}= & \mathbb{E}[\widetilde{W}_{k,it}|\widetilde{W}_{k,it}>0]\mathbb{P}[\widetilde{W}_{k,it}>0]d_{L}\\
 & =\mathbb{E}[\widetilde{W}_{k,it}]d_{L},
\end{aligned}
\]
 which holds by the law of iterated expectations. Finally, for $B_{4}$:
\[
B_{4}=\mathbb{E}[\widetilde{W}_{k,it}]\mathbb{P}[\widetilde{W}_{k,it}>0]d_{L}.
\]
 Hence: $\ensuremath{B_{1}-B_{2}+B_{3}+B_{4}=\mathbb{E}[\widetilde{W}_{k,it}^{2}]-\mathbb{E}[\widetilde{W}_{k,it}]^{2}=var(\widetilde{W}_{k,it})}$,
which means that $\ensuremath{\int_{d_{L}}^{d_{U}}w_{1}(\lambda)d\lambda+w_{0}=1}$.
\end{proof}
\begin{proof}
\textbf{Proof of Theorem \ref{thm:(PVARs-estimate-average-ATE-ACR)}.
}The Theorem is immediate from nesting Definition \ref{def:(ATTs)}
and Theorem \ref{thm:(PVARs-estimate-weighted-average)}. 
\end{proof}
\begin{proof}
\textbf{Proof of Theorem \ref{thm:(PVAR:-Randomised,-dummy,}.} The
proof follows from the application of Lemma \ref{lemmaa1}. 
\end{proof}
\begin{proof}
\textbf{Proof of Theorem \ref{th:(identification-of-dummy-disturbances}.
}Notice that in the case of heterogeneous treatments, by Lemma \ref{lemmaa1}
$\hat{\gamma}_{jk}=\widetilde{Y}_{j,it}(\widetilde{W}_{k,it}=1)-\widetilde{Y}_{j,it}(\widetilde{W}_{k,it}=0)$
estimates
\[
\hat{\gamma}=\frac{\sum_{t=1}^{T}(\widetilde{W}_{it}-\overline{W}_{i})(\widetilde{Y}_{it}-\overline{Y}_{i})}{\sum_{t=1}^{T}(\widetilde{W}_{it}-\overline{W}_{i})^{2}}.
\]

Here $\overline{W}_{i}=\frac{1}{T}\sum_{t=1}^{T}\widetilde{W}_{it}$
and $\overline{Y}_{i}=\frac{1}{T}\sum_{t=1}^{T}\widetilde{Y}_{it}$.
If $\widetilde{W}_{it}$ is either equal to $1$ or $0$,
this is identical to the difference in average outcomes of $\hat{\gamma}=\widetilde{Y}_{it}(\widetilde{W}_{it}=1)-\widetilde{Y}_{it}(\widetilde{W}_{it}=0).$
\[
\begin{aligned}\gamma_{jk}= & \mathbb{E}[\widetilde{Y}_{j,it}|t\in T_{P},i\in I_{P}]-\mathbb{E}[\widetilde{Y}_{j,it}|t\in T_{P},i\in I_{C}]\\
 & -\mathbb{E}[\widetilde{Y}_{j,it}|t\in T_{C},i\in I_{P}]+\mathbb{E}[\widetilde{Y}_{j,it}|t\in T_{C},i\in I_{C}].
\end{aligned}
\]
 Expressing it in terms of potential outcomes: 
\[
\begin{aligned}\mathbb{E}[\widetilde{Y}_{j,it}(\widetilde{W}_{i,1:t-1},1,\widetilde{W}_{-k,it},\widetilde{W}_{i,t+1:T})|t\in T_{P},i\in I_{P}]\\
-\mathbb{E}[\widetilde{Y}_{j,it}(\widetilde{W}_{i,1:t-1},0,\widetilde{W}_{-k,it},\widetilde{W}_{i,t+1:T})|t\in T_{P},i\in I_{C}]\\
-\mathbb{E}[\widetilde{Y}_{j,it}(\widetilde{W}_{i,1:t-1},0,\widetilde{W}_{-k,it},\widetilde{W}_{i,t+1:T})|t\in T_{C},i\in I_{P}]\\
+\mathbb{E}[\widetilde{Y}_{j,it}(\widetilde{W}_{i,1:t-1},0,\widetilde{W}_{-k,it},\widetilde{W}_{i,t+1:T})|t\in T_{C},i\in I_{C}]
\end{aligned}
\]
 Moving from the definition of realized outcomes to the one of potential
outcomes: 
\[
\begin{aligned}\gamma_{jk}= & \mathbb{E}[\widetilde{Y}_{j,it}(1)|t\in T_{P},i\in I_{P}]-\mathbb{E}[\widetilde{Y}_{j,it}(0)|t\in T_{P},i\in I_{C}]\\
 & -\mathbb{E}[\widetilde{Y}_{j,it}(0)|t\in T_{C},i\in I_{P}]+\mathbb{E}[\widetilde{Y}_{j,it}(0)|t\in T_{C},i\in I_{C}]
\end{aligned}
\]
 Adding and subtracting the unobserved potential outcome of the treated
units in treated periods $\mathbb{E}[\widetilde{Y}_{j,it}(0)|t\in T_{P},i\in I_{P}]$

\[
\begin{aligned}\begin{aligned}\gamma_{jk}=\end{aligned}
 & \mathbb{E}[\widetilde{Y}_{j,it}(1)|t\in T_{P},i\in I_{P}]-\mathbb{E}[\widetilde{Y}_{j,it}(0)|t\in T_{P},i\in I_{P}]\\
\mathbb{} & -\mathbb{E}[\widetilde{Y}_{j,it}(0)|t\in T_{P},i\in I_{C}]+\mathbb{E}[\widetilde{Y}_{j,it}(0)|t\in T_{P},i\in I_{C}]\\
 & -\mathbb{E}[\widetilde{Y}_{j,it}(0)|t\in T_{C},i\in I_{P}]+\mathbb{E}[\widetilde{Y}_{j,it}(0)|t\in T_{C},i\in I_{C}]
\end{aligned}
\]
 By assumption \ref{assu:(Parallel-trends):} the second line cancels
out, by assumption \ref{assu:(No-anticipation)} the third line cancels
out. Hence, 
\[
\gamma_{jk}=\mathbb{E}[\widetilde{Y}_{j,it}(1)-\widetilde{Y}_{j,it}(0)|t\in T_{p},i\in I_{P}]
\]
Which can be rewritten into 
\[
\gamma_{jk}=\mathbb{E}[\widetilde{Y}_{j,it}(1)-\widetilde{Y}_{j,it}(0)|\widetilde{W}_{k,it}=1].
\]
\end{proof}
\begin{proof}
\textbf{Proof of Theorem \ref{th:spillover-fail}. } 
\[
\begin{aligned}
\gamma_{jk} &= \mathbb{E}[\widetilde{Y}_{j,k,it}(1,s)\mid t\in T_{p},i\in I_{p};S=s] 
- \mathbb{E}[\widetilde{Y}_{j,k,it}(0,s)\mid t\in T_{p},i\in I_{c};S=s] \\
&\quad - \mathbb{E}[\widetilde{Y}_{j,k,it}(0,0)\mid t\in T_{c},i\in I_{p}] 
+ \mathbb{E}[\widetilde{Y}_{j,k,it}(0,0)\mid t\in T_{c},i\in I_{c}]
\end{aligned}
\]

which then, by the absence of spillover effects in control periods, becomes
\[
\begin{aligned}
\gamma_{jk} &= \mathbb{E}[\widetilde{Y}_{j,k,it}(1,s)\mid t\in T_{p},i\in I_{p};S=s] 
- \mathbb{E}[\widetilde{Y}_{j,k,it}(0,s)\mid t\in T_{p},i\in I_{c};S=s]
\end{aligned}
\]

adding and subtracting the counterfactual 
$\mathbb{E}[\widetilde{Y}_{j,k,it}(0,0)\mid t\in T_{p},i\in I_{p};S=s]$, it becomes

\[
\begin{aligned}
\gamma_{jk} &= \mathbb{E}[\widetilde{Y}_{j,k,it}(1,s)\mid t\in T_{p},i\in I_{p};S=s] 
- \mathbb{E}[\widetilde{Y}_{j,k,it}(0,s)\mid t\in T_{p},i\in I_{c};S=s] \\
&\quad + \mathbb{E}[\widetilde{Y}_{j,k,it}(0,0)\mid t\in T_{p},i\in I_{p};S=s] 
- \mathbb{E}[\widetilde{Y}_{j,k,it}(0,0)\mid t\in T_{p},i\in I_{p};S=s]
\end{aligned}
\]

which can be grouped as 
\[
\begin{aligned}
\gamma_{jk} &= \mathbb{E}[\widetilde{Y}_{j,k,it}(1,s)\mid t\in T_{p},i\in I_{p};S=s] 
- \mathbb{E}[\widetilde{Y}_{j,k,it}(0,0)\mid t\in T_{p},i\in I_{p};S=s] \\
&\quad + \mathbb{E}[\widetilde{Y}_{j,k,it}(0,0)\mid t\in T_{p},i\in I_{p};S=s] 
- \mathbb{E}[\widetilde{Y}_{j,k,it}(0,s)\mid t\in T_{p},i\in I_{c};S=s]
\end{aligned}
\]

and by the definitions of $\widetilde{\text{ATTE}}_{jk}$ and $\widetilde{\text{ASTE}}_{jk}$ becomes 

\[
\gamma_{jk} = \widetilde{\text{ATTE}}_{jk} - \widetilde{\text{ASTE}}_{jk}.
\]

\end{proof}
\begin{proof}
\textbf{Proof of Theorem \ref{th:spillover-identify}. } Notice that the estimator $\delta_{j,k}$ in equation \ref{eq: estimation-spillover} is
\[
\begin{aligned}
\delta_{j,k} = 
&\ \mathbb{E}[Y_{j,k,it}(1,s) \mid t \in T_{p}, i \in I_{p}; S_{j,k,it} = s] 
- \mathbb{E}[Y_{j,k,it}(0,0) \mid t \in T_{p}, i \in I_{c}; S_{j,k,it} = 0] \\
& - \mathbb{E}[Y_{j,k,it}(0,0) \mid t \in T_{c}, i \in I_{p}; S_{j,k,it} = 0] 
+ \mathbb{E}[Y_{j,k,it}(0,0) \mid t \in T_{c}, i \in I_{c}; S_{j,k,it} = 0]
\end{aligned}
\]

Notice that, using assumption \ref{assu:modified-no-ancicip}:
\[
\begin{aligned}
\delta_{j,k} = 
\mathbb{E}[Y_{j,k,it}(1,s) \mid t \in T_{p}, i \in I_{p}; S_{j,k,it} = s] 
- \mathbb{E}[Y_{j,k,it}(0,0) \mid t \in T_{p}, i \in I_{c}; S_{j,k,it} = 0]
\end{aligned}
\]

Further, by adding and subtracting $\mathbb{E}[Y_{j,k,it}(0,0) \mid t \in T_{p}, i \in I_{p}; S_{j,k,it} = s]$:
\[
\begin{aligned}
\delta_{j,k} = 
&\ \mathbb{E}[Y_{j,k,it}(1,1) \mid t \in T_{p}, i \in I_{p}; S_{j,k,it} = 1] 
- \mathbb{E}[Y_{j,k,it}(0,0) \mid t \in T_{p}, i \in I_{c}; S_{j,k,it} = 0] \\
& + \mathbb{E}[Y_{j,k,it}(0,0) \mid t \in T_{p}, i \in I_{p}; S_{j,k,it} = s] 
- \mathbb{E}[Y_{j,k,it}(0,0) \mid t \in T_{p}, i \in I_{p}; S_{j,k,it} = s]
\end{aligned}
\]
and, by assumption \ref{assu:modified-parallel-trends}:
\[
\delta_{j,k} = 
\mathbb{E}[Y_{j,k,it}(1,s) - Y_{j,k,it}(0,0) \mid t \in T_{p}, i \in I_{p}; S_{j,k,it} = s] 
\]
which is equal to the total effect on the treated by definition.
\end{proof}

\clearpage{}
\end{document}